\documentclass[reprint, amsmath, amssymb, aps, pra]{revtex4-2}
 
\usepackage{graphicx}
\usepackage{multirow}
\usepackage[justification=raggedright]{caption}
\usepackage{subcaption}
\usepackage{dcolumn}
\usepackage{bm}
\usepackage{braket}
\usepackage{siunitx}
\usepackage[colorlinks]{hyperref}
\hypersetup{allcolors=blue}
\usepackage{changes}
\usepackage[ruled,vlined,linesnumbered]{algorithm2e}
\usepackage{algorithmic}
\usepackage{verbatim}

\usepackage{cleveref}

\usepackage{amssymb,amsmath,amsthm,amsfonts}

\newtheorem{theorem}{Theorem}
\newtheorem{corollary}{Corollary}
\newtheorem{lemma}{Lemma}
\newtheorem{proposition}{Proposition}
\newtheorem{definition}{Definition}

\newcommand{\eqn}[1]{(\ref{eqn:#1})}
\newcommand{\eq}[1]{(\ref{eq:#1})}
\newcommand{\thm}[1]{\hyperref[thm:#1]{Theorem~\ref*{thm:#1}}}
\newcommand{\cor}[1]{\hyperref[cor:#1]{Corollary~\ref*{cor:#1}}}
\newcommand{\defn}[1]{\hyperref[defn:#1]{Definition~\ref*{defn:#1}}}
\newcommand{\lem}[1]{\hyperref[lem:#1]{Lemma~\ref*{lem:#1}}}
\newcommand{\prop}[1]{\hyperref[prop:#1]{Proposition~\ref*{prop:#1}}}
\newcommand{\assum}[1]{\hyperref[assum:#1]{Assumption~\ref*{assum:#1}}}
\newcommand{\fig}[1]{\hyperref[fig:#1]{Figure~\ref*{fig:#1}}}
\newcommand{\sfig}[1]{\hyperref[sfig:#1]{Supplementary Fig.~\ref*{sfig:#1}}}
\newcommand{\tab}[1]{\hyperref[tab:#1]{Table~\ref*{tab:#1}}}
\newcommand{\algo}[1]{\hyperref[algo:#1]{Algorithm~\ref*{algo:#1}}}
\renewcommand{\sec}[1]{\hyperref[sec:#1]{Section~\ref*{sec:#1}}}
\newcommand{\append}[1]{\hyperref[app:#1]{Supplementary Note~\ref*{app:#1}}}
\newcommand{\fac}[1]{\hyperref[fac:#1]{Fact~\ref*{fac:#1}}}
\newcommand{\lin}[1]{\hyperref[lin:#1]{Line~\ref*{lin:#1}}}
\newcommand{\prob}[1]{\hyperref[prob:#1]{Problem~\ref*{prob:#1}}}
\newcommand{\conj}[1]{\hyperref[conj:#1]{Conjecture~\ref*{conj:#1}}}
\newcommand{\exam}[1]{\hyperref[exam:#1]{Example~\ref*{exam:#1}}}

\def\>{\rangle}
\def\<{\langle}

\newcommand{\R}{\mathbb{R}}

\newcommand{\M}{\mathcal{M}}

\newcommand{\eps}{\epsilon}

\DeclareMathOperator{\Var}{Var}

\DeclareMathOperator{\diag}{diag}

\DeclareMathOperator{\Haf}{Haf}
\DeclareCaptionLabelSeparator{bar}{\ \textbar\ }

\renewcommand{\emptyset}{\varnothing}
\def\:{\hbox{\bf:}}
\def\mix{\operatorname{mix}}
\def\TV{\operatorname{TV}}

\newcommand{\blue}[1]{\textcolor{blue}{#1}}
\DontPrintSemicolon

\begin{document}
\title{Efficient Classical Sampling from Gaussian Boson Sampling Distributions on Unweighted Graphs}

\author{Yexin Zhang$^{1,2}$}
\thanks{Equal Contribution.}

\author{Shuo Zhou$^{1,2}$}
\thanks{Equal Contribution.}

\author{Xinzhao Wang$^{1,2}$}
\thanks{Equal Contribution.}

\author{Ziruo Wang$^{1,2}$,\\Ziyi Yang$^{1,2,3}$}

\author{Rui Yang$^{1,2}$}

\author{Yecheng Xue$^{1,2}$}

\author{Tongyang Li$^{1,2}$}

\thanks{Corresponding author. Email: tongyangli@pku.edu.cn}

\affiliation{\vspace{1em}\textsuperscript{1}Center on Frontiers of Computing Studies, Peking University, Beijing 100871, China}

\affiliation{\textsuperscript{2}School of Computer Science, Peking University, Beijing 100871, China}

\affiliation{\textsuperscript{3}School of Mathematical Sciences, Peking University, Beijing 100871, China}

\begin{abstract} 
Gaussian Boson Sampling (GBS) is a promising candidate for demonstrating quantum computational advantage and can be applied to solving graph-related problems. In this work, we propose Markov chain Monte Carlo-based algorithms to sample from GBS distributions on undirected, unweighted graphs. Our main contribution is a double-loop variant of Glauber dynamics, whose stationary distribution matches the GBS distribution. We further prove that it mixes in polynomial time for dense graphs using a refined canonical path argument. 
Numerically, we conduct experiments on unweighted graphs with 256 vertices, larger than the scales in former GBS experiments as well as classical simulations. In particular, we show that both the single-loop and double-loop Glauber dynamics improve the performance of original random search and simulated annealing algorithms for the max-Hafnian and densest $k$-subgraph problems up to 10$\times$. Overall, our approach offers both theoretical guarantees and practical advantages for efficient classical sampling from GBS distributions on unweighted graphs.
\end{abstract}

\maketitle


\section{Introduction}

Recent years have witnessed increasing efforts to demonstrate quantum computational advantage over classical computers using real quantum devices~\cite{Arute_2019,Zhong_2020}. In particular, Gaussian Boson Sampling (GBS), implemented with quantum photonics, has shown promise with experimental demonstrations achieving advantages over classical methods~\cite{Zhong_2020,PhysRevLett.127.180502,PhysRevLett.131.150601,madsen2022quantum,deshpande2022quantum}. GBS has been applied to graph-related problems, such as the densest $k$-subgraph problem~\cite{arrazola2018using} and max-Hafnian calculations~\cite{arrazola2018quantum} through numerical simulations. These applications are based on the encoding of a graph’s adjacency matrix into a Gaussian state, where the probability of measuring a specific photon-number pattern is proportional to the squared Hafnian of the corresponding submatrix~\cite{hamilton2017gaussian,PhysRevA.100.032326,bradler2018gaussian}. More recently,~\citet{PhysRevLett.130.190601} demonstrated GBS on a noisy quantum device, enhancing classical algorithms for solving graph problems.

On the other hand, classical simulation algorithms for GBS have been actively explored. Quesada and Arrazola~\cite{Quesada_2020} introduced an exact classical algorithm for simulating GBS by sequentially sampling the number of photons in each mode conditioned on the previously sampled modes, which runs in polynomial space and exponential time. \citet{oh2022classical} developed a classical algorithm for Boson sampling based on dynamic programming, which exploits the limited connectivity of the linear-optical circuit to improve efficiency. Separately, \citet{Oh_2024} proposed a tensor-network-based classical algorithm to simulate large-scale GBS experiments with photon loss, requiring relatively modest computational resources.  In the context of graph-theoretical applications of GBS, an important  property is that the adjacency matrix encoded in the Gaussian Boson sampler is non-negative, which is believed to make the problem more tractable than in the general case. The aforementioned result~\cite{Quesada_2020} reduced the simulation of GBS with non-negative matrices to the problem of estimating Hafnians of non-negative matrices, which can be done efficiently for the adjacency matrix of certain strongly expanding graphs \cite{rudelson2016hafnians}. \citet{PRXQuantum.5.020341} also took advantage of this non-negativity property to design a quantum-inspired classical algorithm for finding dense subgraphs and their numerical results suggest that the advantage offered by a Gaussian Boson sampler is not significant. However, an open question remains whether a classical algorithm for GBS on general graphs with provable performance guarantees can achieve a computational cost comparable to that of a Gaussian Boson sampler.

In this work, we adopt Markov chain Monte Carlo (MCMC) algorithms to sample from GBS distributions on unweighted graphs. MCMC is a standard class of sampling algorithms with well-established theoretical guarantees~\cite{jerrum2003counting,levin2017markov}. Among MCMC methods, Glauber dynamics~\cite{glauber1963time} is particularly widespread due to its simplicity and rigorous analytical foundations. Glauber dynamics generates samples from the matchings of an undirected and unweighted graph by iteratively adding or removing edges with biased transition probabilities. The resulting stationary distribution is proportional to the power of the number of edges in the matchings. For each matching sampled from Glauber dynamics, we consider its support vertex set as a sampled subset of vertices. The probability of sampling a given vertex set is further weighted by the number of perfect matchings within that set, which is also the Hafnian of the adjacency matrix of the graph. 

We propose a double-loop Glauber dynamics with a rigorous theoretical guarantee that its stationary distribution is identical to the sampling distribution of GBS on unweighted graphs. Specifically, unlike the standard single-loop Glauber dynamics, which yields a stationary distribution proportional to the Hafnian of subgraphs, the double-loop approach ensures a stationary distribution proportional to the square of the Hafnian, which coincides with the sampling distribution of GBS. Concretely, in the double-loop Glauber dynamics, when deciding whether to remove an edge, we run a secondary Markov chain to uniformly sample a perfect matching from the current subgraph. The removal decision is then based on whether this edge reappears in the newly sampled matching. Furthermore, for dense graphs, we prove that the double-loop Glauber dynamics has a polynomial mixing time, demonstrating its computational feasibility in these cases. The key to demonstrate the rapid mixing lies in the canonical path technique introduced in~\cite{jerrum2003counting}, which routed flows between every pair of matchings without creating particularly congested ``pipes'', and bounded the mixing time by estimating the maximum congestion of possible transitions. More specifically, for complete graphs, we establish an enhanced canonical path framework that permits multiple paths between any two matchings. By constructing specially designed paths with favorable symmetry properties to distribute flows efficiently, we ultimately estimate maximum congestion through calculating total congestion of symmetrical transitions.
For dense graphs, the maximum congestion is bounded by estimating the ratio between the congestion of dense graphs and complete graphs, thereby yielding a polynomial mixing time. 
This result is particularly significant because the dense graph represents the regime where classical methods are challenging. Sparse graphs often permit classical shortcuts that dense graphs lack. For instance, the Maximum Clique problem, a canonical task that GBS has been proposed to solve \cite{banchi2020molecular}, is NP-hard in general. However, its classical complexity is greatly reduced on sparse graphs, where algorithms can exploit structural properties such as low degeneracy to find maximal cliques efficiently \cite{eppstein2013listing}.

Our numerical simulations confirm that both single-loop and double-loop Glauber dynamics improve the performance of original random search and simulated annealing algorithms for the max-Hafnian and densest $k$-subgraph problems, providing empirical validations of our theoretical findings. Our experiments are conducted on unweighted graphs with 256 vertices, larger than the scales in former GBS experiments~\cite{PhysRevLett.130.190601} as well as classical simulations~\cite{Oh_2024}. In verification, the variants enhanced by Glauber dynamics are up to 3$\times$ better than the original classical algorithms. On random graphs, the enhanced variants achieve score advantage up to 4$\times$. On bipartite graphs, the enhanced variants are up to 10$\times$ better than the original classical algorithms. 
 
The rest of the paper is organized as follows. In \sec{GBS-defn}, we review the definition of GBS. We introduce the standard Glauber dynamics for sampling matchings in \sec{Glauber}, and then propose our double-loop Glauber dynamics for sampling from GBS distributions on unweighted graphs with provable guarantee in \sec{double-loop-Glauber}. We present all experimental results in \sec{experiments}.


\section{Results}
\subsection{Gaussian Boson Sampling for graph problems}\label{sec:GBS-defn}
Boson sampling is a quantum computing model where $N$ identical photons pass through a $M$-mode linear interferometer and are detected in output modes~\cite{10.1145/1993636.1993682}. In the standard Boson sampling paradigm, the probability of a given output configuration $\bar{n}$ is related to the permanent of a submatrix of the interferometer’s $M\times M$ unitary matrix $T$, which we call $T_S$:
\begin{align}
\operatorname{Pr}(\bar{n})=\left|\operatorname{Perm}\left(T_S\right)\right|^2=\left|\sum_{\sigma \in \mathcal{S}_N} \prod_{i=1}^N T_{S_{i, \sigma(i)}}\right|^2,
\end{align}
where $\mathcal{S}_N$ is the set of all permutations on $[N]:=\{1,2,\ldots,N\}$, and $T_S$ is a matrix composed of the intersecting elements of the columns and the rows of $T$  determined by the input positions and output $\bar{n}$, respectively.

Gaussian Boson Sampling (GBS) is a variant that uses Gaussian states with squeezing parameters $\{r_i\}_{i=1}^M$ as inputs instead of single photons. In GBS, the output photon-number distribution is determined by a matrix function called the Hafnian. The Hafnian of a $2n\times2n$ matrix $A$ is defined as 
\begin{align}\label{eq:Haf}
\operatorname{Haf}(A)=\frac{1}{2^n n!} \sum_{\sigma \in \mathcal{S}_{2 n}} \prod_{i=1}^n A_{\sigma(2 i-1), \sigma(2 i)}.
\end{align}
Specifically, the Hafnian of the adjacency matrix of an unweighted graph equals to the number of its perfect matchings~\cite{bradler2018gaussian}.
The probability of measuring a specific photon number pattern $\bar{n} = (n_1, n_2, \dots, n_M)$ in an $M$-mode GBS experiment can be expressed in closed-form as~\cite{hamilton2017gaussian,PhysRevA.100.032326}
\begin{align}\label{eq:prob}
\operatorname{Pr}(\bar{n})=\frac{1}{n_1!n_2!\cdots n_M!\sqrt{\det(\sigma+\mathbb{I}_{2 M} / 2)}} \operatorname{Haf}({A}_S),
\end{align}
where $\sigma$ is the $2M\times 2M$ covariance matrix of the Gaussian state and ${A}_S$ is the submatrix by selecting the intersection of columns and rows only according to output $\bar{n}$ from the sampling matrix ${A}={B}\oplus {B^*}$ with
\begin{align}\label{eqn:squeezing}
    {B} = T \diag \{\tanh r_1,\tanh r_2,\dots,\tanh r_M\} T^{\top}.
\end{align}

Given an arbitrary undirected graph with potentially complex-valued symmetric adjacency matrix $\Delta$, we aim to engineer $B=c\Delta$ with an appropriate rescaling parameter $c$. It is possible to find such a $T$ by the Takagi-Autonne decomposition (see~\cite{doi:10.1139/cjp-2024-0070} and Section 4.4 of~\cite{horn2012matrix}) when $0<c<1/({\max_j|\lambda_j|})$, where $\{\lambda_j\}$ is the eigenvalue set of $A$.\
Subsequently, the sampling matrix becomes $A=c \Delta\oplus c\Delta^*$. 

When $N=O(\sqrt{M})$, with dominating probability all the click number $n_i\leq 1$~\cite{10.1145/1993636.1993682}. Then, all the factorials in \eqref{eq:prob} disappear since $0!=1!=1$. On the other hand, the covariance matrix $\sigma$ and the sampling matrix $A$ is related by~\cite{PhysRevA.100.032326}
\begin{align}
    (\sigma+\mathbb{I}_{2 M} / 2)^{-1}=\mathbb{I}_{2 M}-\left(\begin{array}{ll}
0_M & \mathbb{I}_M \\
\mathbb{I}_M & 0_M
\end{array}\right)A.
\end{align}
Thus, the probability of outputting the subgraph with vertex set $S$ is given by 
\begin{align}
\operatorname{Pr}_{A}(S)& =\sqrt{\det[(\sigma+\mathbb{I}_{2 M} / 2)^{-1}]} \operatorname{Haf} [(c\Delta\oplus c\Delta^*)_S] \\
& = \sqrt{\det \left(\begin{array}{cc}
\mathbb{I}_M & -c\Delta^* \\
-c\Delta & \mathbb{I}_M
\end{array}\right)} \operatorname{Haf}(c\Delta_S)\cdot \operatorname{Haf}(c\Delta_S)^{*}\\
& = \prod_{j=1}^{2M} \sqrt{1-c^2 \lambda_j^2} c^{2|S|} |\operatorname{Haf} (\Delta_S)|^2.
\end{align}

In all, the task of sampling from GBS distributions on a graph with real-valued adjacency matrix is equivalent to developing algorithms that sample a subgraph with probability proportional to
\begin{align}
    \operatorname{Pr}(\Delta_S)\propto c^{2|S|} \operatorname{Haf}^2 (\Delta_S).
\end{align}
We remark that as a special case of the Hafnian, the permanent of a matrix with non-negative entries admits probabilistic polynomial-time approximation, in contrast to the \#P-hardness for general complex matrices \cite{10.1145/1993636.1993682, valiant1979complexity, jerrum2004polynomial}. This suggests that the applications of GBS for max-Hafnian and densest $k$-subgraph of nonnegative-weight graphs may be simpler than the complex-valued version, and efficient classical sampling algorithms from GBS distributions on such instances are worth investigation.

\subsection{Glauber dynamics for matching}\label{sec:Glauber}

Our algorithm is built upon the Glauber dynamics, a well-established Markov chain Monte Carlo method for sampling. In particular, the Glauber dynamics that samples across the space of all matchings of a graph is known as the monomer-dimer model. Given a graph $G=(V,E)$ and a fugacity parameter $\lambda>0$, we denote $\M$ to be the collection of all the matchings of $G$. We define the Gibbs distribution $\mu$ for the monomer-dimer model as $\mu(X) = w(X)/Z$ for $\forall X\in \M$, where the weight $w(X) = \lambda^{|X|}$ and $Z$ is a normalizing factor known as the partition function. (This partition function is known as the matching polynomial, which has rich study in MCMC literature~\cite{barvinok2016combinatorics}. The matching polynomial is closely related to the loop hafnian, a variant of the Hafnian that can be used to count all matchings in a graph \cite{10.1145/3325111,PhysRevA.105.052412}.)

In a step $t$ when the Glauber dynamics is at a matching $X_t\in \M$, it chooses an edge $e$ uniformly at random from $E$. If $X'\oplus \{e\}$ forms a new matching (either a larger matching or $e$ is in $X'$), then we let $X_{t+1} = X'$ with probability $w(X')/(w(X')+w(X_t))$ and otherwise let $X_{t+1} = X_t$. If $X'$ and $e$ do not form a new matching, simply set $X_{t+1} = X_t$. This Glauber dynamics for matchings is formally presented in \algo{glauber}. 

\begin{algorithm}[h]
        \SetAlgoLined 
        \KwIn{A graph $G=(V,E)$, number of steps $T$.}
        \KwOut{A sample of matching $X$ of $G$ such that
            $\Pr[X] \propto \lambda^{|X|}$. }
\vspace{1em}            
            Initialize $X_0$ as an arbitrary matching in $G$;\;
            Initialize $t\leftarrow 0$;\;
        
        \lWhile{$t < T$}\
            \Indp Choose a uniformly random edge $e$ from $E$;\; 
            \lIf{$e$ and $X_t$ form a new matching}\
            {\Indp Set $X_{t+1} = X_t \cup \{e\}$ with probability $\frac{\lambda}{1+\lambda}$ and otherwise set $X_{t+1} = X_t$;\;}
            
            \lElse 
            {\lIf{$e$ is in $X_t$}\
             {\Indp Set $X_{t+1} = X_t \backslash \{e\}$ with probability $\frac{1}{1+\lambda}$ and otherwise set $X_{t+1} = X_t$;\label{lin:glauber-8}\;
             \Indm\lElse\  
              { \quad \quad Set $X_{t+1} =X_t$;}}
              
             $t \leftarrow t+1$;}
     
     \Indm Output the matching $X_T$;\;
     \caption{Glauber Dynamics for Matchings}
     \label{algo:glauber}
\end{algorithm}

It is straightforward to verify that the Markov chain is ergodic, aperiodic, and irreducible, meaning that it converges to a unique stationary distribution~\cite{levin2017markov}. 
We can verify that the Glauber dynamics converges to the Gibbs distribution by checking that the detailed balance condition holds: 
For two matchings $X$ and $X\cup\{e\}$, we have
\begin{align}
    \Pr[X \to X\cup\{e\}] = \frac{\lambda}{1+\lambda}  \frac{1}{|E|},\\ \Pr[X\cup\{e\}\to X] = \frac{1}{1+\lambda}  \frac{1}{|E|}, 
\end{align}
and thus
\begin{align}
    \mu(X) \Pr[X \to X\cup\{e\}] = \mu(X\cup\{e\}) \Pr[X\cup\{e\} \to X].
\end{align}

Furthermore, the convergence speed of an MCMC to its stationary distribution is characterized by its mixing time, defined as the time required by the Markov chain to have sufficiently small distance to the stationary distribution. Formally, let $P_t(X_0,\cdot)$ denote the distribution of matchings after $t$ steps starting from $X_0$. The total variation distance between $P_t(X_0,\cdot)$ and the stationary distribution $\mu$ is defined as
\begin{align}
    \|P_t(X_0,\cdot) - \mu\|_{\TV} := \frac{1}{2} \sum_{X} |P_t(X_0,X) - \mu(X)|.
\end{align}
Thus we can define the mixing time of the Markov chain:
\begin{align}
    t_{\mix}(\eps)&:= \min\left\{t:\max_{X_0\in\M} \|P_t(X_0,\cdot) - \mu\|_{\TV}\leq \eps\right\},\\
    t_{\mix} &:= t_{\mix}({1}/{4}).
\end{align}

It is known that the Glauber dynamics for matchings has a polynomial mixing time:

\begin{theorem}[\cite{jerrum2003counting}]\label{thm:glauber-mixing}
    For a general graph $G$ with $n$ vertices and $m$ edges, the mixing time of the Glauber dynamics for the monomer-dimer model on $G$ with fugacity $\lambda>0$ is $O(n^2 m \log n)$. 
\end{theorem}

Now, we examine the Gibbs distribution from an alternative perspective. For a sampled matching $X$, we take its vertex set $S = V(X)$ as the final output, and let $\nu$ denote the stationary distribution of vertex sets. Since the number of perfect matchings on $S$ is given by the Hafnian of the subgraph $G_S$ induced by $S$, which we briefly write as $\Haf(S)$, the stationary probability of $S$ satisfies the following facts if we take $\lambda=c^2$:
\begin{align}
    \Pr[X] &\propto \lambda^{|X|}= c^{2|X|}=c^{|S|},\\
    \Pr[S] &\propto c^{|S|}\Haf(S). \label{eqn:Glauber-single-power}
\end{align}

Similarly, we can denote $P'_t(X_0,\cdot)$ as the distribution of vertex sets after $t$ steps starting from $X_0$, and define the total variation distance between $P_t'(X_0,\cdot)$ and the stationary distribution $\nu$ as
\begin{align}
    \|P_t'(X_0,\cdot) - \nu\|_{\TV} := \frac{1}{2} \sum_{S} |P'_t(X_0,S) - \nu(S)|.
\end{align}

The mixing time for vertex set sampling is defined as
\begin{align}
    t_{\mix}^S(\eps)&:= \min\left\{t:\max_{X_0\in\M} \|P'_t(X_0,\cdot) - \nu\|_{\TV}\leq \eps\right\},\\
    t_{\mix}^S &:= t_{\mix}^S({1}/{4}).
\end{align}

Since
\begin{align}
    &\|P'_t(X_0,\cdot) - \nu\|_{\TV} \\
    =& \frac{1}{2} \sum_{S} |P'_t(X_0,S) - \nu(S)|\\
    \leq& \frac{1}{2} \sum_{S} \sum_{X\text{ is a perfect matching of }S} |P_t(X_0,X) - \mu(X)|\\
    =& \|P_t(X_0,\cdot) - \mu\|_{\TV},
\end{align}
the mixing time of the Glauber dynamics for matchings is at most the mixing time of the Glauber dynamics for vertex sets.

We note that the distribution in Eq.~\eqn{Glauber-single-power} resembles the GBS distribution: The Gibbs distribution involves the Hafnian to the first power, whereas the GBS distribution weights each vertex set by the square of its Hafnian. This quadratic dependence amplifies the probability mass on larger vertex sets with numerous perfect matchings, resulting in a more concentrated distribution on such vertex sets that potentially gives favorable solutions to problems such as densest $k$-subgraph~\cite{arrazola2018using} and max-Hafnian~\cite{arrazola2018quantum}.


\subsection{Double-loop Glauber dynamics}
\label{sec:double-loop-Glauber}

Inspired by the standard Glauber dynamics, we develop several enhanced algorithms that achieve classical sampling from GBS distributions on unweighted graphs, i.e., sampling of a vertex set $S$ with distribution $\Pr(S)\propto c^{2|S|}\Haf^2(S)$. 

A simple idea is rejection sampling, which is a basic technique applied to generate samples from a target distribution by sampling from a proposal distribution and accepting or rejecting the samples based on a certain criterion. For sampling from GBS distributions, our rejection sampling algorithm works as follows:
\begin{itemize}
    \item Run two instances of the Glauber dynamics for matchings independently to sample two vertex sets $S_1,S_2$ with probability $\Pr[S]\propto {c}^{|S|}\Haf(S)$.
    \item Accept $S_1$ if $S_1 = S_2$, otherwise reject and repeat the process.
\end{itemize}
In this rejection sampling algorithm, the probability of accepting a vertex set $S$ is
\begin{align}
    \Pr[S\text{ is accepted}] &= \Pr[S_1 = S]\cdot \Pr[S_2 = S] \\
    &\propto ({c}^{|S|}\Haf(S))^2=c^{2|S|}\Haf^2(S).
\end{align}

However, this method is inefficient as it may incur a large number of rejections. For general graphs, since the number of possible vertex subsets grows exponentially, if no vertex subset dominates the weight distribution, the acceptance probability for each rejection sampling attempt can be exponentially small.

In this work, we propose a novel double-loop Glauber dynamics that directly samples from the distribution $\Pr[S] \propto c^{2|S|}\Haf^2(S) $, where $S$ is a vertex set of a graph $G$ and $c$ is a constant. Note that it is equivalent to realizing the sampling of matchings according to the distribution $\Pr[X]\propto (c^2)^{2|X|}\Haf(G_X)$ , where $G_X$ denotes the subgraph induced by the vertex set of the matching $X$. 

In contrast to the standard Glauber dynamics for matchings (\algo{glauber}), our approach introduces modified transition probabilities for edge removal, carefully calibrated to ensure the convergence to the desired stationary distribution. These probabilities are dynamically determined through an auxiliary inner Markov chain that operates in each step of the Glauber dynamics and samples a perfect matching from the subgraph induced by the current matching. Specifically, different from \lin{glauber-8} in \algo{glauber}, when the random edge $e$ is in the current matching $X_t$, we first uniformly sample a perfect matching $E_t$ in the subgraph $G_{X_t}$ induced by $X_t$. If $e$ is not in $E_t$, we keep the current matching $X_t$. Otherwise, we remove $e$ from $X_t$ with probability $1/(1+\lambda^2)$ and otherwise keep $X_t$. Our algorithm is formally presented in \algo{double-loop} with an illustration in \fig{double-loop}.

\begin{algorithm}[h]
    \SetAlgoLined 
	\KwIn{A graph $G=(V,E)$.}
	\KwOut{A sample of vertex set $S$ of $G$ such that
        $\Pr[S] \propto \lambda^{|S|}\Haf^2(S)$. }

        \vspace{1em}
        Initialize $X_0$ as an arbitrary matching in $G$;\
        
        Initialize $t\leftarrow 0$, set $T = \tilde{O}(n^6)$;\
	
	\lWhile{$t < T$}\
    {
        \Indp Choose a uniformly random edge $e$ from $E$; \label{lin:double-loop-4}\

        \lIf{$e$ and $X_t$ form a new matching}\
        {\Indp Set $X_{t+1} = X_t\cup\{e\}$ with probability $\frac{\lambda^2}{1+\lambda^2}$ and otherwise set $X_{t+1} = X_t$;\label{lin:double-loop-6}\;}
        
        \lElse 
        {\lIf{$e$ is in $X_t$}\
         {\Indp Uniformly sample a perfect matching $E_t$ in the subgraph induced by $X_t$ (by running another MCMC on $G_{X_t}$). If $e\not\in E_t$, set $X_{t+1} = X_t$. If $e\in E_t$, set $X_{t+1} = X_t \backslash \{e\}$ with probability $\frac{1}{1+\lambda^2}$ and otherwise set $X_{t+1} = X_t$;\label{lin:double-loop-8}\; \label{lin:inner}
         \Indm\lElse\
          {\quad \quad Set $X_{t+1} =X_t$;}}}
        $t \leftarrow t+1$;\;
        \caption{Double-loop Glauber Dynamics}
    \label{algo:double-loop}
 }
 Output the vertex set in matching $X_T$;\;
\end{algorithm}

\begin{figure*}[!htbp]
    \centering
        \includegraphics[width=\linewidth]{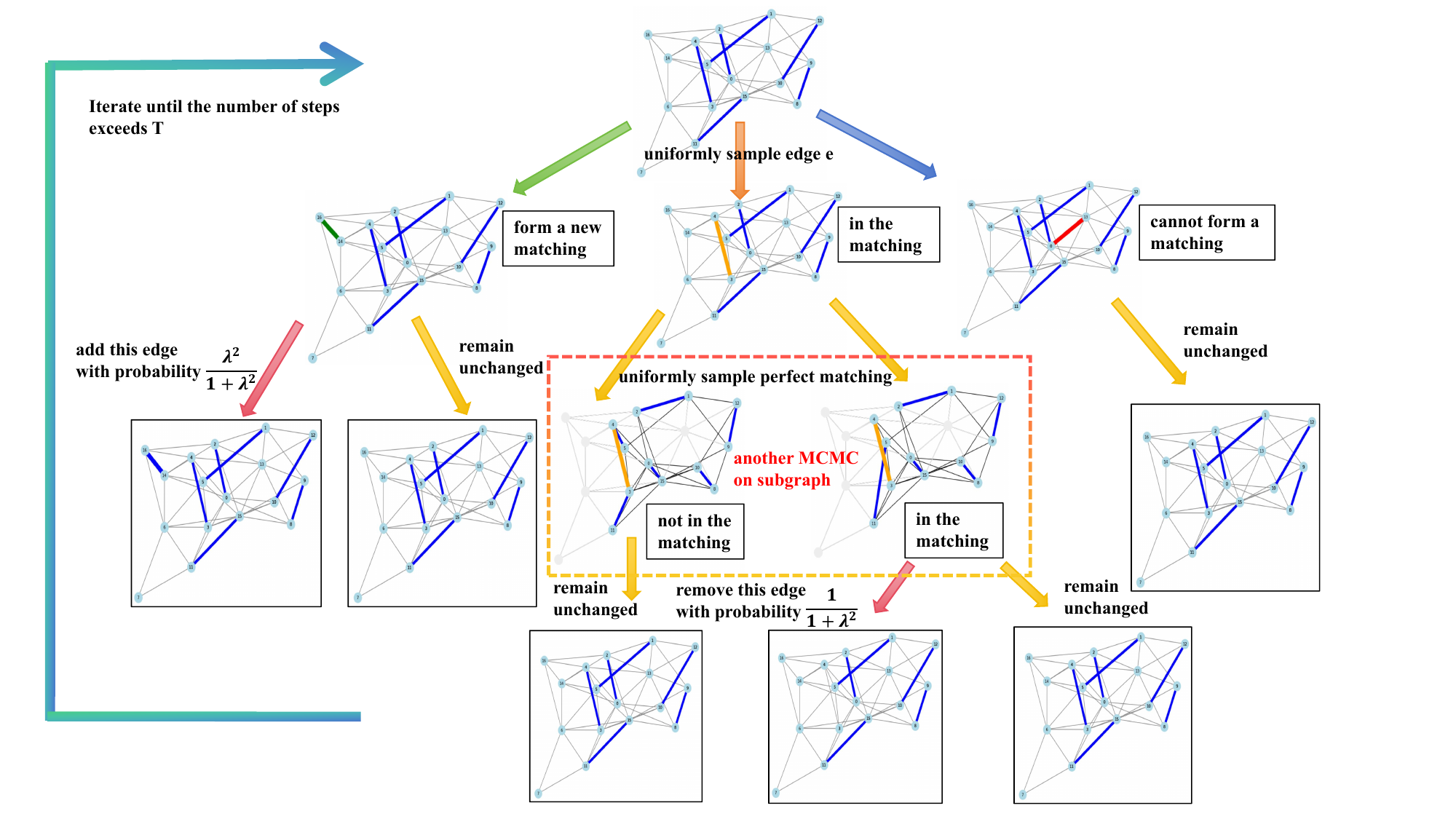}
        \caption{Flowchart of the Double-loop Glauber Dynamics.}
        \label{fig:double-loop}
\end{figure*}

We first verify the convergence to the desired distribution. For the edge sets of any two matchings $X$ and $X\cup\{e\}$, on the one hand, from \lin{double-loop-6} in \algo{double-loop},
\begin{align}
    \Pr[X\to X\cup\{e\}] = \frac{\lambda^2}{1+\lambda^2} \cdot \frac{1}{|E|}.
\end{align}
On the other hand, if $X_t = X\cup\{e\}$, the probability $\Pr[e\in E_t]$ in \lin{double-loop-8} in \algo{double-loop} is equal to the ratio of the number of perfect matchings of $G_{X\cup\{e\}}$ that contains $e$ to the number of perfect matchings of $G_{X\cup\{e\}}$. The numerator equals the number of perfect matchings of $G_{X}$. Thus we have
\begin{align}
    \Pr[X\cup\{e\} \to X] = \frac{1}{|E|} \cdot \frac{\Haf(G_X)}{\Haf(G_{X\cup\{e\}})} \cdot \frac{1}{1+\lambda^2},
\end{align}
where the second term comes from the probability of edge $e$ being in $E_t$. Recall that the stationary distribution of matchings should satisfy
\begin{align}
    \pi(X) \Pr[X\to X\cup\{e\}] = \pi(X\cup\{e\}) \Pr[X\cup\{e\}\to X],
\end{align}
thus if we take $\lambda=c^2$,
\begin{align}
    \pi(X) \propto \lambda^{2|X|} \Haf(G_X) = c^{4|X|} \Haf(G_X).
\end{align}
Notice that the number of perfect matchings of $G_X$ is $\Haf(G_X)$. Therefore, when we finally take the vertex set of the sampled matching as the output, the probability of the sampled set of vertices $S$ is:
\begin{align}
    \Pr[S] \propto c^{2|S|} \Haf^2(S),
\end{align}
which simulates the output distribution by GBS. We remark that our \algo{double-loop} works for any undirected and unweighted graph $G$. 

Next, we rigorously establish that the mixing time of the double-loop Glauber dynamics on dense graphs is at most a polynomial. Specifically, for dense bipartite graphs, we have:
\begin{theorem}\label{thm:bipartite-graph}
    Given a bipartite graph $G=(V_1,V_2,E)$ with $|V_1|=m$, $|V_2|=n$, and $m\geq n$. If the minimum degree of vertices in $V_1$ satisfies $\delta(V_1) \geq n - \xi$ and the minimum degree of vertices in $V_2$ satisfies $\delta(V_2) \geq m - \xi$ for some constant $\xi$, then for $\lambda>\frac{1}{4}$, the mixing time of the double-loop Glauber is polynomially bounded in $M$ and $n$, specifically $\tilde{O}(m^2n^{2\xi+4})$.
\end{theorem}
For dense non-bipartite graphs, we have:
\begin{theorem}\label{thm:non-bipartite-graph}
    Given a non-bipartite graph $G=(V,E)$ with $|V|=2n$, if the minimum degree of $G$ satisfies $\delta(V) \geq 2n - \xi$ for some constant $\xi$, for $\lambda>\frac{1}{4}$, the mixing time of the double-loop Glauber is polynomially bounded in $n$, specifically $\tilde{O}(n^{2\xi+6})$.
\end{theorem}

Our analysis primarily employs a proof technique in MCMC literature known as the canonical path method~\cite{jerrum2003counting}, which establishes mixing time bounds by constructing a proper multicommodity flow problem, and then selecting suitable transition paths between states. For a pair of initial state $I$ and final state $F$, we can conceptualize the problem as routing $\pi(I)\pi(F)$ units of distinguishable flow from state $I$ to state $F$, utilizing the Markov chain's transitions as ``pipes''. We can define an arbitrary canonical path from $I$ to $F$ for each pair $I,F\in\Omega$, and the corresponding ``congestion'' as follows:
\begin{align}\label{eq:congestion1}
\kern-2mm\varrho := \max_{t=(M_1,M_2)}\left\{\frac{\sum_{I,F:\gamma \text{ uses }t} \pi(I)\pi(F) |\gamma_{IF}|}{\pi(M_1)P(M_1,M_2)} \right\}
\end{align}
where $\gamma_{IF}$ is the path from $I$ to $F$, and $|\gamma_{IF}|$ denotes the length of $\gamma_{IF}$.  Ref.~\cite{jerrum2003counting} proved that the mixing time of Markov chain is bounded by 
\begin{align}
    t_{\mix}(\eps) = O\left(\varrho\left( \ln {\eps^{-1}} + \ln{\pi_{\min}(M)}^{-1}\right)\right).
\end{align}

The proof of \thm{glauber-mixing} in~\cite{jerrum2003counting} applied the canonical path method, in which the path from matching $I$ to matching $F$ is defined by decomposing $I\otimes F$ into a collection of paths and even-length cycles, and then processing these components with some specific order. The approximation of \eq{congestion1} is achieved by constructing an injective mapping from $(I,F)$ to another matching for each transition. 
However, the canonical path method is not directly applicable to the double-loop Glauber dynamics, as the inner Markov chain (\lin{inner}) introduces additional Hafnian terms to the transition process, leading to a multiplicative error term that grows exponentially in the worst case when using original injection constructions. To address this challenge, we develop an alternative proof technique that leverages symmetries between different transitions.

In the original canonical path method, if the graph has sufficient symmetry, we can directly compute the congestion of each transition instead of constructing an injective mapping to bound the congestion. 
We begin our analysis with the complete graph case, focusing on direct computation of the congestion for each individual transition. We can design special canonical path to ensure that for any transition originating from a matching of size $k$, only paths transitioning between smaller matchings $(|M| < k)$ and larger matchings $(|M| > k)$ will utilize this transition. This insight enables us to collectively compute the total congestion across all transitions originating from size $k$ matchings. Furthermore, we aim to design symmetric canonical paths such that the congestion values for symmetric transitions are the same in complete graphs. This symmetry enables us to compute the congestion of individual transitions by summing the contributions from all paths that pass through them.
However, the original canonical path construction, which determines the order of components of $I\oplus F$ based on fixed vertex orderings, fails to preserve the necessary symmetry properties if the graph is not complete. To relax this limitation, we introduce an enhanced canonical path framework that permits multiple distinct paths between each state pair $(I, F)$. Specifically, our solution involves considering all possible permutations of the connected components, systematically constructing transformation sequences for each ordering (for more details, see \append{canonical-path}).

For complete graphs and complete bipartite graphs, this symmetric construction leverages the matching enumeration properties specific to each graph, enables precise congestion calculations and ultimately yields a polynomial mixing time bound through careful analysis of the path distribution and transition probabilities. 
The complete technical proof is provided in \append{complete-bipartite} (for bipartite graphs) and \append{complete-non-bipartite} (for non-bipartite graphs). 
For sufficiently dense graphs where the Hafnian of each subgraph differs from the complete graph case by at most a polynomial factor, our analysis naturally extends to establishing polynomial mixing time bounds, as formalized in \thm{bipartite-graph} and \thm{non-bipartite-graph}. The complete technical proof is provided in \append{dense-bipartite} (for bipartite graphs) and \append{dense-non-bipartite} (for non-bipartite graphs).

Another crucial aspect of our framework is the implementation of the inner Markov chain for uniform sampling of perfect matchings in subgraphs. Ref.~\cite{jerrum2004polynomial} introduced an efficient algorithm that achieves an efficient approximation of uniform sampling of perfect matchings of an arbitrary balance bipartite graph:
\begin{lemma}[\cite{jerrum2004polynomial}]\label{lem:bigraph-uniform}
    Given a balance bipartite graph $G=(V_1,V_2,E)$ with $|V_1|=|V_2|=n$, then their exists an algorithm that achieves a uniform sampling of perfect matching of $G$ in time $O(n^{11}(\log n)^2(\log n + \log \eta^{-1}))$, with failure probability $\eta$. 
\end{lemma}
For non-bipartite graphs, Ref.~\cite{jerrum1989approximating} provided a polynomial-time algorithm that approximates uniform sampling of perfect matchings in a dense graph:

\begin{lemma}[\cite{jerrum1989approximating}]\label{lem:non-bigraph-uniform}
    Given a non-bipartite graph $G=(V,E)$ with $|V|=2n$, if the minimum degree of vertices in $V$ satisfies $\delta(V)\geq n$, then their exists an algorithm that achieves a uniform sampling of perfect matching of $G$ in time $ \tilde{O}(n^{14}(\ln {\eta^{-1}})^2)$, with failure probability $\eta$. 
\end{lemma}
We briefly discuss these approximate uniform sampling methods for perfect matchings in \append{uniform-sampling-pm}.

By integrating the outer Markov chain framework and the inner uniform sampling mechanism along with comprehensive error analysis, we establish our main theoretical results. On dense bipartite graphs, we have:

\begin{theorem}\label{thm:main-results-bigraph} 
    Given a bipartite graph $G=(V_1,V_2,E)$ with $|V_1|=m, |V_2|=n$ and $m\geq n$. If the minimum degree of vertices in $V_1$ satisfies $\delta_1(G) \geq n - \xi$, and the minimum degree of vertices in $V_2$ satisfies $\delta_2(G) \geq m - \xi$ for some constant $\xi$, then for $\lambda>\frac{1}{4}$, given error $\eps$, we can achieve a sampling in time $\tilde{O}{(m^2n^{(15+2\xi)}(\log\eps^{-1})^2)}$ such that the total variation distance between the sampling distribution and the ideal stationary distribution is at most $\epsilon$.  
\end{theorem}

On dense non-bipartite graphs, we have:
\begin{theorem}\label{thm:main-results-non-bigraph} 
    Given a non-bipartite graph $G=(V,E)$ with $|V|=2n$, If the minimum degree of vertices in $V$ satisfies $\delta_1(G) \geq 2n - 1 - \xi$, for some constant $\xi$, then for $\lambda>\frac{1}{4}$,
    given error $\eps$, we can achieve a sampling in time $\tilde{O}(n^{2\xi+20}(\log \epsilon^{-1})^3)$ such that the total variation distance between the sampling distribution and the ideal stationary distribution is at most $\epsilon$. 
\end{theorem}
We note that our theoretical performance bound in Theorem \ref{thm:main-results-bigraph} is stronger for bipartite graphs. This advantage stems from the fact that the subroutine of uniformly sampling perfect matchings is known to be classically more efficient for bipartite graphs \cite{jerrum2004polynomial}. For context, the BipartiteGBS technique also leverages the unique properties of bipartite graphs, though for the different goal of encoding arbitrary matrices for hardness proofs~\cite{Grier_2022}.


\subsection{Numerical experiments}\label{sec:experiments}
We conduct experiments to compare our algorithms with prior approaches. All results and plots are obtained by numerical simulations on a 10-core Apple M2 Pro chip with 16 GB memory and an NVIDIA L4 chip with 24 GB GPU memory via python 3.11.0. We use `networkx' library~\cite{hagberg2008exploring} to store and manipulate graphs and `thewalrus' library~\cite{Gupt2019} to calculate Hafnian from adjacency matrix. Specifically, the graph problems we aim to solve are defined as follows:
\begin{itemize}
    \item Max-Hafnian: Given an undirected graph with non-negative adjacent matrix and target subgraph size $k$, find the subgraph of size $k$ with the maximum Hafnian value defined as \eq{Haf}.
    \item Densest $k$-subgraph: Given an undirected graph with non-negative adjacent matrix and target subgraph size $k$, find the subgraph of size $k$ with the maximum density value. Density denotes the number of edges divided by the number of vertices.
\end{itemize}

We design various unweighted graphs and set their total number of vertices $n=256$, which is larger than the system size of the 144-mode quantum photonic device \textit{Ji\v{u}zhāng}~\cite{PhysRevLett.130.190601} used in the GBS experiment solving graph problems as well as classical simulations. Specifically, the graphs we use are presented as follows: 

\begin{itemize}
    \item $G_1$ for max-Hafnian: The edge set contains a 16-vertex complete graph and all remaining vertex pairs have an edge with probability $0.2$. As a result, finding the 16-vertex subgraph with maximal Hafnian essentially finds the 16-vertex complete graph.
    
    \item $G_2$ for densest $k$-subgraph: Vertex $i$ has edges to vertices $0,~1,~…,~n-1-i$, which means the degree of each vertex is decreasing. As a result, finding the 80-vertex densest $k$-subgraph essentially finds the induced subgraph of the first 80 vertices.
    
    \item $G_3$ for score advantage: An Erd\"{o}s-Rényi graph as one of the most representative  random graphs. For each vertex pair there exists an edge with fixed probability 0.4.
    
    \item $G_4$ for double-loop Glauber dynamics: A random bipartite graph with 128 vertices in each part, and there is an edge on each vertex pair between two parts with fixed probability 0.2.

    \item $G_5$ for sparse graph experiments: A random bipartite graph with 128 vertices in each part, and there are $10\cdot 256$ edges chosen among all the $128^2$ possible pairs uniformly at random.
\end{itemize}

The classical algorithms and their variants enhanced by Glauber dynamics are presented in detail in \append{algorithm}. We empirically choose sufficient and appropriate mixing time and post-select the edge sets with right size from the iterative process of Glauber dynamics. 

In addition to \algo{glauber} and \algo{double-loop} introduced in the previous sections, we also implement another typical MCMC algorithm by Jerrum~\cite{jerrum2003counting} (\algo{jerrum-glauber}) and a quantum-inspired classical algorithm by~\citet{PRXQuantum.5.020341} for comparison.

\begin{algorithm}[h]
        
        \KwIn{A graph $G=(V,E)$, number of steps $T$.}
        \KwOut{A sample of matching $X$ of $G$ such that
            $\Pr[X] \propto \lambda^{|X|}$.}
\vspace{1em}
            Initialize $X_0$ as an arbitrary matching in $G$; 

            Initialize $t\leftarrow 0$;
        
        \lWhile{$t < T$}\
            \Indp Choose a uniformly random edge $e$ from $E$;
            
            \lIf{$e$ and $X_t$ form a new matching}\
            {\quad \quad Set $M=X_t\cup\{e\}$;}
            
            \lElse 
            {\lIf{$e$ is in $X_t$}\
             {\Indp \hspace{.1em} Set $M=X_t\backslash\{e\}$;

            \Indm \lElse {\lIf {$e$ has a common vertex $v$ with $e'$ in $X_t$ and the remaining vertex $u$ is not covered by $X_t$}\
            \Indp \hspace{.1em} Set $M=X_t\cup \{e\} \backslash \{e'\};$}
            
             \Indm\lElse\  
              {\quad \quad Set $M =X_t$;}}}

             Set $X_{t+1}=M$ with prob.~$\min\{1,\lambda^{|M|-|X|}\}$ and otherwise set $X_{t+1}=X_t$;
             
             $t \leftarrow t+1$;

     \Indm Output the matching $X_T$;

    \caption{Jerrum's Approach: Glauber Dynamics for Matchings~\cite{jerrum2003counting}} 
     
    \label{algo:jerrum-glauber}
\end{algorithm}

\subsubsection{Glauber dynamics verification and comparison}\label{sec:5A}
We first verify that the Glauber dynamics can enhance random search and simulated annealing by outperforming their original algorithms as corresponding baselines. We plot our experimental results in \fig{2}. The instance $G_1$ serves as a test on the classic planted clique problem \cite{alon1998finding}, while $G_2$ acts as a sanity check with a predictable density structure. The objective is to confirm that our algorithms can correctly identify optimal solutions.

We note that the quantum-inspired classical algorithm in~\citet{PRXQuantum.5.020341} requires collision-free outcomes for a given target click number, which appears with very low probability when $k$ is large. Therefore, it is not applied to the densest $k$-subgraph problem in our experiments. This to some extent demonstrates the advantage of the Glauber dynamics for problem settings with an arbitrarily high click number $k$, since adding or decreasing edges poses no limitation on collision.

It can be seen from \fig{2} that compared to the original random search and simulated annealing algorithms, substituting uniform update in each iteration with \algo{glauber}, \algo{jerrum-glauber}, and the quantum-inspired classical algorithm all significantly improve their performance in terms of Hafnian and density values up to 3$\times$. These results confirm the correctness of our methods and their effectiveness in locating ground-truth solutions.

\begin{figure*}[!htbp]
    \centering
    \includegraphics[width=\linewidth]{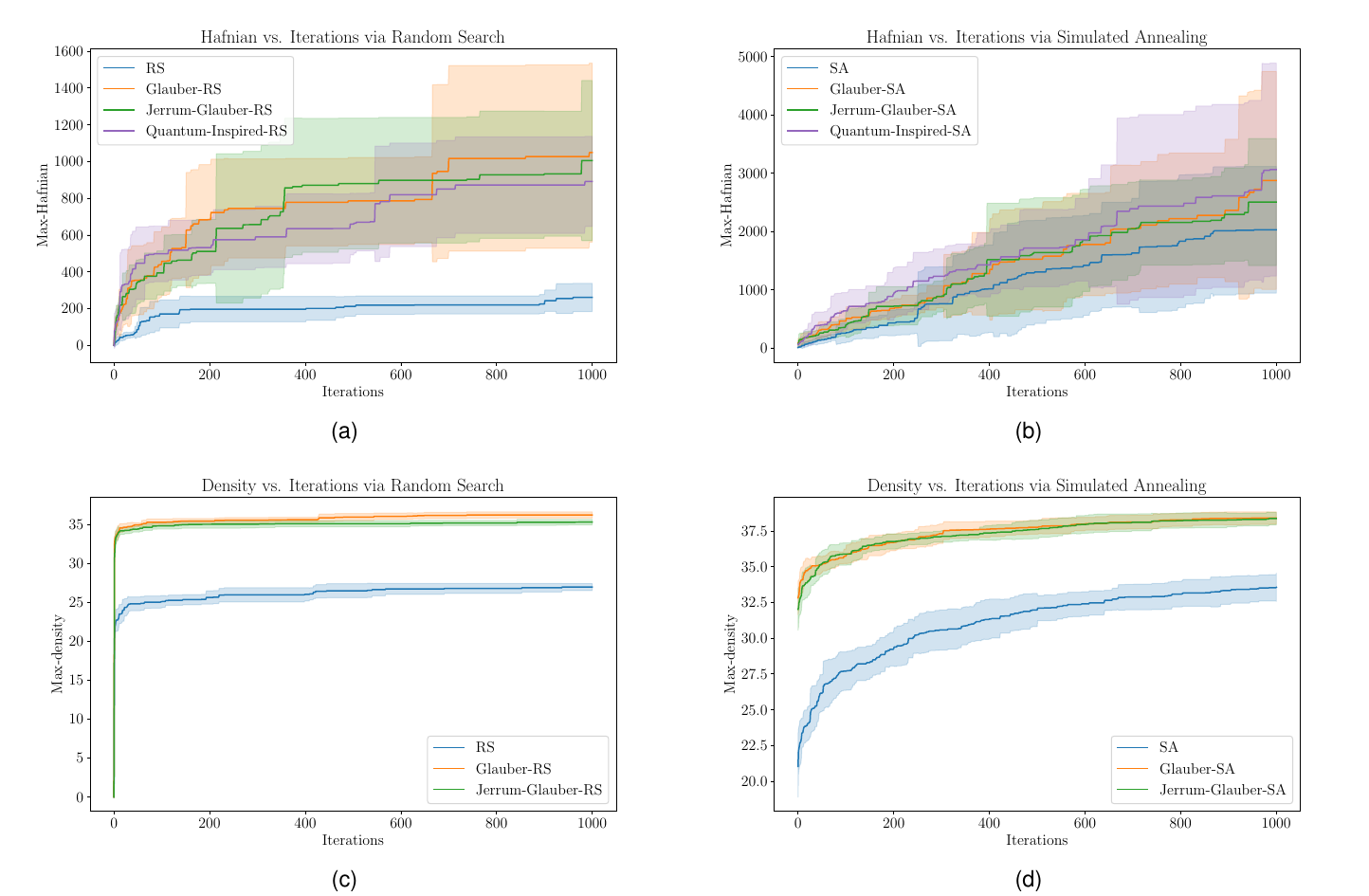}
    \caption{Glauber Dynamics Verification and Comparison. Iterations $=1000$, mixing time $\leq 10000$, fugacity $\lambda=c^2$, annealing parameter $\gamma=0.95$. (a) Hafnian Random Search on $G_1$ with $k=16$, $c=0.1$. The three enhanced variants are on average $240-300\%$ higher than the original classical algorithm, and the upper confidential intervals are about 350\% better. (b) Hafnian Simulated Annealing on $G_1$ with $k=16$, $c=0.1$. The three enhanced variants are on average $20-50\%$ higher than the original classical algorithm, and the upper confidential intervals are about 55\% better. (c) Density Random Search on $G_2$ with $k=80$, $c=0.4$. The two enhanced variants are on average $30-35\%$ higher than the original classical algorithm, and the upper confidential intervals are about 32\% better. (d) Density Simulated Annealing on $G_2$ with $k=80$, $c=0.4$. The two enhanced variants are on average $10-15\%$ higher than the original classical algorithm, and the upper confidential intervals are about 12\% better.}
    \label{fig:2}
\end{figure*}

\subsubsection{Score advantage comparison with GBS}
To evaluate performance on a more challenging average-case scenario, we use the Erd\"{o}s-Rényi random graph $G_3$. The score advantage is defined as the ratio of the maximum Hafnian/density value acquired by random search enhanced by the double-loop Glauber dynamics to that acquired by the original random search at the end of the stage with 100 iterations. We remark that too many iterations may result in collapse of score advantage to 1 since even the original approach can acquire the same optimal solution as the variants by the Glauber dynamics. We plot our experimental results enhanced by \algo{double-loop} in \fig{3}, whose probability distribution on unweighted graph is comparable to GBS in~\cite{PhysRevLett.130.190601}.

Ref.~\cite{PhysRevLett.130.190601} used a 144-vertex complete graph with random complex weights, whereas our numerical experiment is conducted on the aforementioned 256-vertex Erd\"{o}s-Rényi graph $G_3$ with $p=0.4$. Since all weights of the edges in the graph are real, we do not need the square of modulus. These may explain the difference between our score advantage and Fig. 3 of~\cite{PhysRevLett.130.190601}. As shown in the left side of our \fig{3}, our Hafnian score advantage up to 4$\times$ is lower than the GBS experiments~\cite{PhysRevLett.130.190601}, whereas in the right side of \fig{3} our density score advantage is comparable to theirs. The key message is that our method is effective and robust for finding dense, high-Hafnian subgraphs within generic, average-case random graphs.

\begin{figure*}[!htbp]
    \centering
    \includegraphics[width=\linewidth]{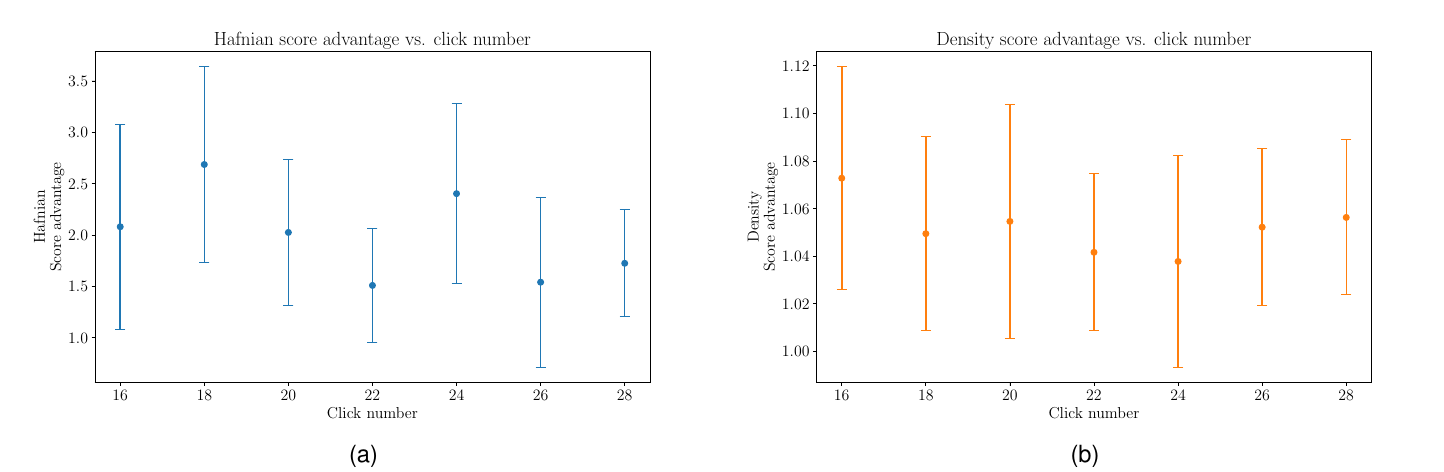}
    \caption{Score advantage vs.~click number: RS enhanced by \algo{double-loop} / original RS. Iterations $=100$, mixing time $\leq1000$, fugacity $\lambda=c^2$. (a) The left side reflects the Hafnian score advantage up to 4$\times$ on $G_3$ with $c=0.6$. (b) The right side reflects the density score advantage up to 120\% on $G_3$ with $c=0.6$.}
    \label{fig:3}
\end{figure*}

\begin{figure*}[!htbp]
    \centering
    \includegraphics[width=\linewidth]{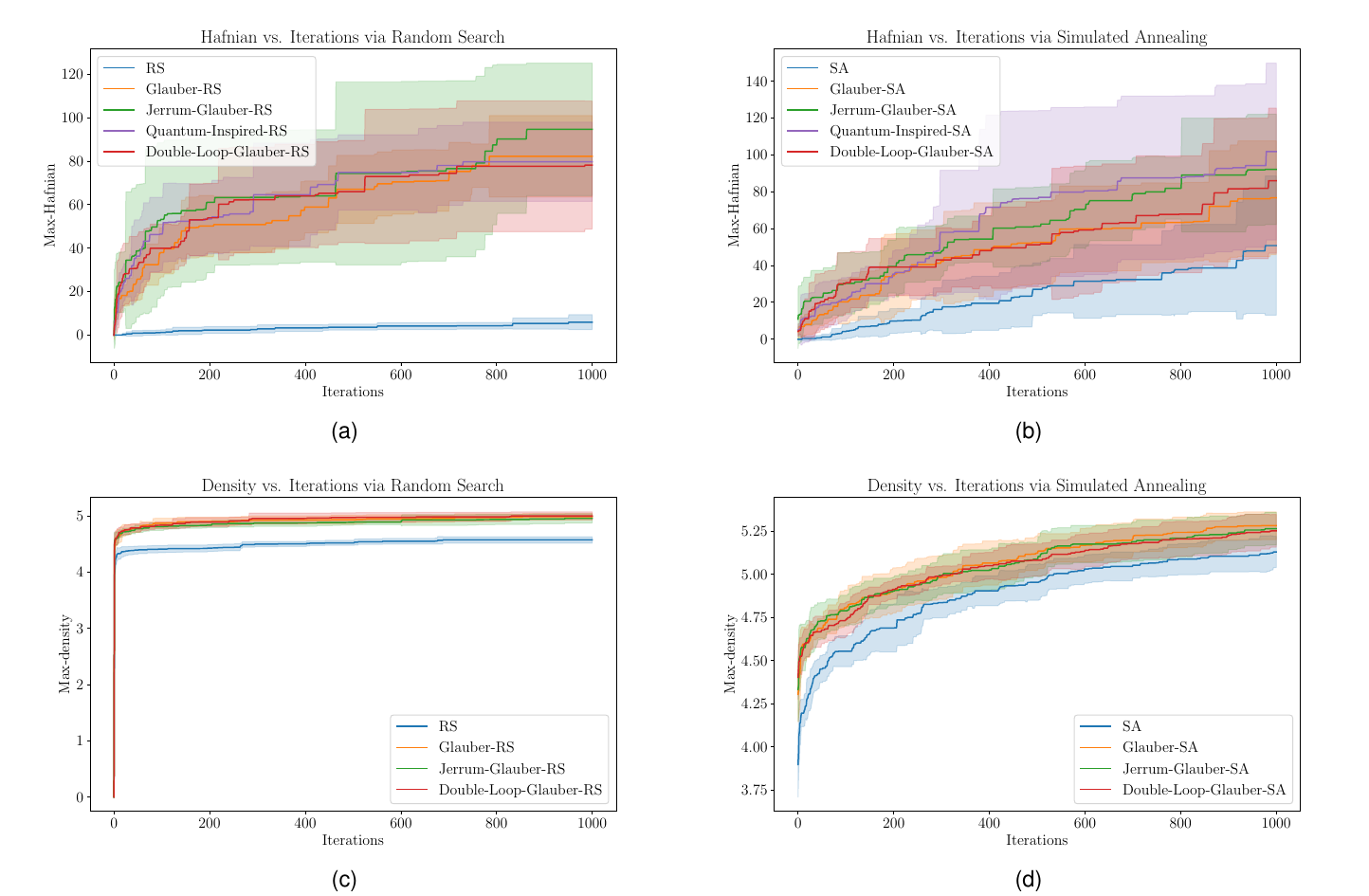}
    \caption{Double-loop Glauber Dynamics. Iterations $=1000$, mixing time $\leq 1000$, fugacity $\lambda=c^2$, annealing parameter $\gamma=0.95$. (a) Hafnian Random Search on $G_4$ with $k=16$, $c=0.4$. The four enhanced variants are on average $12\times-15\times$ higher than the original classical algorithm, and the upper confidential intervals are about 10$\times$ better. (b) Hafnian Simulated Annealing on $G_4$ with $k=16$, $c=0.4$. The four enhanced variants are on average $50-100\%$ higher than the original classical algorithm, and the upper confidential intervals are about 70\% better. (c) Density Random Search on $G_4$ with $k=80$, $c=0.8$. The three enhanced variants are on average 8\% higher than the original classical algorithm, and the upper confidential intervals are about 9\% better. (d) Density Simulated Annealing on $G_4$ with $k=80$, $c=0.8$. The three enhanced variants are on average $3\%$ higher than the original classical algorithm, and the upper confidential intervals are about 2\% better.}
    \label{fig:4}
\end{figure*}

\subsubsection{Double-loop Glauber dynamics on bipartite graphs}
In this section, we further conduct experiments on a randomly connected bipartite graph $G_4$ to verify our theory on double-loop Glauber dynamics in \sec{double-loop-Glauber}. This experiment is crucial because bipartite graphs are common in applications and are structures where our uniform perfect matching sampling subroutine is known to be efficient \cite{jerrum2004polynomial}. For the similar reason stated in \sec{5A}, we do not apply the quantum-inspired classical algorithm in~\cite{PRXQuantum.5.020341} to the densest $k$-subgraph problem. We plot our experimental results in \fig{4}.

It can be seen from \fig{4} that compared to the original random search and simulated annealing algorithms, substituting the uniform update in each iteration with \algo{glauber}, \algo{double-loop}, \algo{jerrum-glauber}, and the quantum-inspired classical algorithm all comparably improve their performance in terms of Hafnian and density values up to 10$\times$. This substantial improvement, especially from the \algo{double-loop} variant, validates our theoretical findings in \sec{double-loop-Glauber}.

\begin{figure*}[!htbp]
    \centering
    \includegraphics[width=\linewidth]{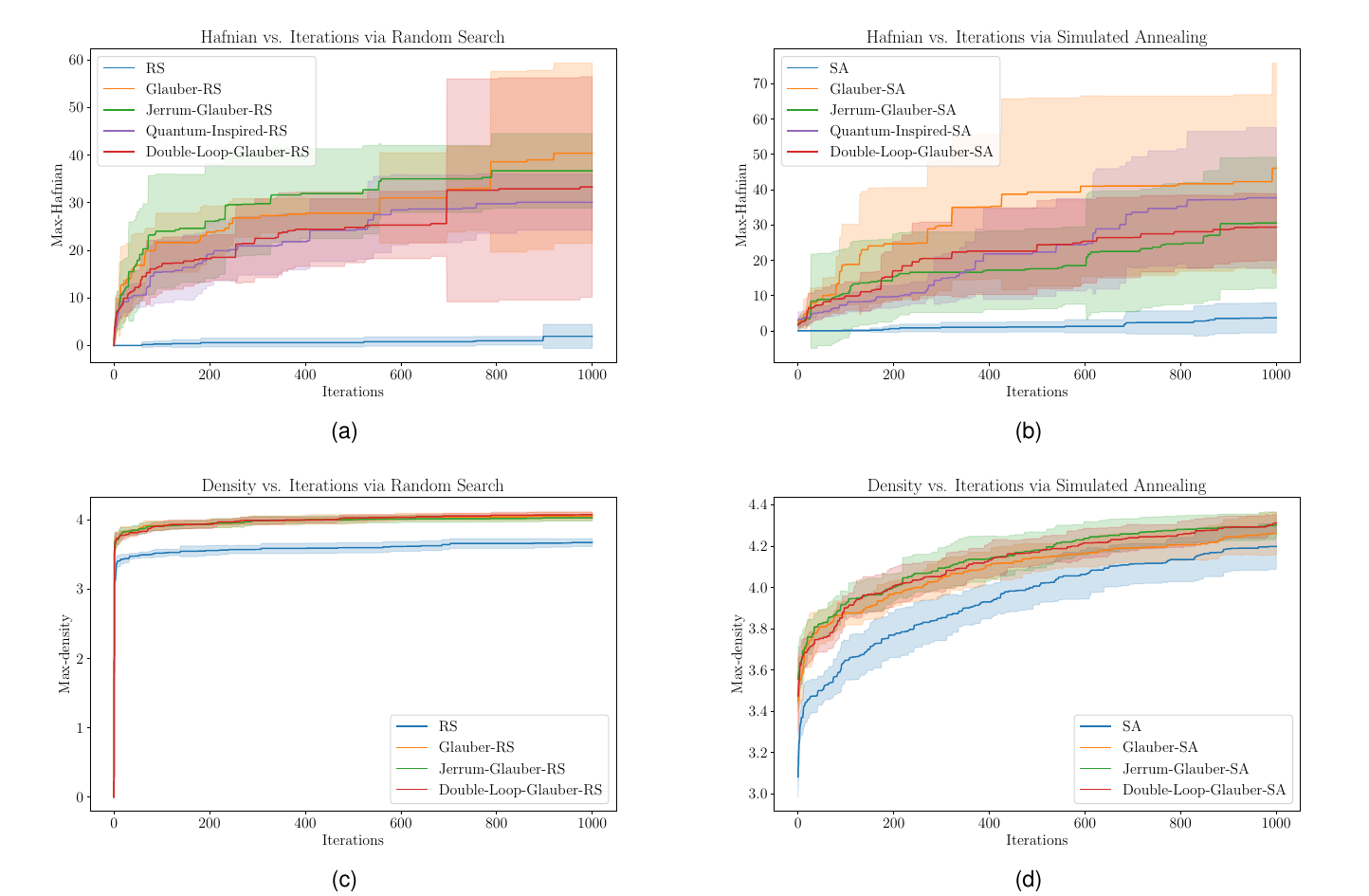}
    \caption{Sparse Graph Experiments. Iterations $=1000$, mixing time $\leq1000$, fugacity $\lambda=c^2$, annealing parameter $\gamma=0.95$. (a) Hafnian Random Search on $G_5$ with $k=16$, $c=0.6$. (b) Hafnian Simulated Annealing on $G_5$ with $k=16$, $c=0.6$. In both cases the four enhanced variants are capable of finding subgraphs with considerable Hafnian values, while the original random search and simulated annealing almost completely fail to find any perfect matching. (c) Density Random Search on $G_5$ with $k=80$, $c=0.8$. (d) Density Simulated Annealing on $G_5$ with $k=80$, $c=0.8$. Regarding the density values, the three enhanced variants are still $10\%$ higher than the performance of original algorithms with the same stopping time as in dense graphs, demonstrating the convergence of our MCMC-based algorithms on this sparse graph.}
    \label{fig:5}
\end{figure*}

\subsubsection{Sparse graph}
In this work, we only provide theoretical guarantee of~\algo{double-loop} for dense graphs in terms of polynomial mixing time. As a complementary from the practical perspective, in this subsection we conduct experiments on a sparse random bipartite graph $G_5$ whose number of edges $m=O(n)$, while for $G_1$, $G_2$, $G_3$, $G_4$ the parameter $m$ scales with $n^2$.

It can be seen from \fig{5} that the original random search and simulated annealing algorithms almost completely fail to find any subgraph with perfect matching and the acquired Hafnian value is nearly zero. In contrast, substituting the uniform update in each iteration with \algo{glauber}, \algo{double-loop}, and \algo{jerrum-glauber}, the classical sampling algorithms are still capable of acquiring considerable Hafnian values.

Regarding the density values, the three enhanced variants are still $10\%$ higher than the performance of original algorithms with the same stopping time as in dense graphs, demonstrating the convergence of our MCMC-based algorithms on this sparse graph.

\section{Discussion}

In this work, we have proposed the double-loop Glauber dynamics to sample from GBS distributions on unweighted graphs and theoretically prove that it mixes in polynomial time on dense graphs. 
Numerically, we conduct experiments on graphs with 256 vertices, larger than the scales in former GBS experiments as well as classical simulations. In particular, we show that both the single-loop and double-loop Glauber dynamics outperform original classical algorithms for the max-Hafnian and densest $k$-subgraph problems with up to 10$\times$ improvement.

Our theoretical results for the rapid mixing time in double-loop algorithms leave several natural directions for future research. 
We have only proved polynomial mixing time results for dense graphs. 
A natural next step is to establish polynomial mixing time bounds for more general graphs 
(e.g., planar graphs, Erd\"{o}s-Rényi graphs, etc.) to better support our empirical observations.  
In particular, we identify the following key challenges for further investigation:
\begin{itemize}
    \item For non-dense subgraphs, the ratio of the Hafnian to the corresponding complete subgraph's Hafnian may become exponentially large, preventing direct application of the complete graph analysis, as is discussed in \append{non-dense-graph}. Can we apply our enhanced canonical path method to deriving mixing time upper bounds for more general graphs, such as graphs with high symmetry?
    
    \item 
    For the analysis of mixing time for non-dense non-bipartite graphs, another technical barrier lies in the uniform sampling of perfect matchings on such graphs, whose mixing time is an open problem and directly impacts the convergence guarantees of our inner-layer MCMC procedure.

    \item Our double-loop MCMC algorithm can also be extended to non-negative weighted graphs by slightly adjusting the transition probabilities of the Markov chain (see \algo{random-search-enhanced} in \append{algorithm}). However, the analysis of their mixing time becomes significantly more complicated and remains an open problem.
\end{itemize}

\section*{Acknowledgments}
We thank Weiming Feng and Yu-Hao Deng for helpful suggestions. We acknowledge the Bohrium platform developed by DP Technology for supporting our numerical experiments.

This work was supported by the National Natural Science Foundation of China (Grant Numbers 92365117 and 62372006).

\section*{Data Availability}
The data and source code are available on our Github page \hyperlink{https://github.com/Qubit-Fernand/GBS-MCMC}{https://github.com/Qubit-Fernand/GBS-MCMC}. In addition, they are also available at Zenodo repository~\cite{shuo_zhou_2025_17046671}. 

\section*{Author contributions}
T.L.~conceived the project. Y.Z.~provided the theoretical proof of the mixing time of MCMC-based algorithms with assistance from Z.Y. 
S.Z.~conducted the numerical experiments with assistance from Z.W., R.Y., and Y.X. 
X.W.~proposed the framework of the double-loop Glauber dynamics. 
All authors contributed to writing the paper. Y.Z., S.Z., X.W., and T.L.~contributed to the review process.

\section*{Competing interests}
The authors declare no competing interests.


%

\appendix
\onecolumngrid

\newpage

\setcounter{figure}{0}
\setcounter{theorem}{0}
\setcounter{algocf}{0}

\renewcommand{\figurename}{Supplementary Fig.}
\renewcommand{\thefigure}{\arabic{figure}}

\captionsetup[figure]{labelfont=bf}
\captionsetup[figure]{labelsep=bar}

\renewcommand{\thetheorem}{S\arabic{theorem}} 
\renewcommand{\thelemma}{S\arabic{lemma}} 
\renewcommand{\thecorollary}{S\arabic{corollary}}
\renewcommand{\thedefinition}{S\arabic{definition}}
\renewcommand{\theproposition}{S\arabic{proposition}}
\renewcommand{\thealgocf}{S\arabic{algocf}}

\renewcommand{\appendixname}{Supplementary Note}
\appendix
\renewcommand{\thesection}{\arabic{section}}
\renewcommand{\thesubsection}{\Alph{subsection}}

\renewcommand{\theequation}{S\arabic{equation}}
\setcounter{equation}{0}

\section*{Supplementary Information}
In this Supplementary Information, we present detailed proofs of the propositions, theorems, and corollaries presented in the manuscript ``Efficient Classical Sampling from Gaussian Boson Sampling Distributions on Unweighted Graphs''. The methods used in our proofs are structured as follows. In \append{canonical-path}, we review the canonical path method, a well-established technique in MCMC literature. \append{uniform-sampling-pm} summarizes existing results on uniform sampling for perfect matchings. In \append{dense-graph-proof}, we present the proof for the polynomial mixing time of our double-loop Glauber dynamics on dense graphs, which is our core contribution. In \append{double-loop-weighted}, we describe how to adapt our double-loop Glauber dynamics to weighted graphs. Finally, we present details of the algorithms in our numerical experiments with pseudocodes in \append{algorithm}.

\section{Canonical Path Method}\label{app:canonical-path}

The key to demonstrating a rapid mixing of a Markov chain using the canonical path technique lies in setting up a suitable multicommodity flow. For a pair $I,F\in \Omega$, we can imagine that by using the transitions of the Markov chain, we hope to route $\pi(I)\times \pi(F)$ units of fluid from $I$ to $F$. Specifically, we can define a canonical path $\gamma_{IF}:= (M_0 = I,M_1,\dots, M_l = F)$ from $I$ to $F$ through pairs $(M_i,M_{i+1})$ of states adjacent in the Markov chain. Let
\begin{align}
    \Gamma := \{\gamma_{IF}\mid I,F\in\Omega\}
\end{align}
be the set of all canonical paths. Then we can define the congestion $\varrho = \varrho(\Gamma)$ of the chain:
\begin{align}\label{eq:congestion}
    \varrho(\Gamma) := \max_{t=(M,M')}\left\{\frac{1}{\pi(M)P(M,M')} \sum_{I,F:\gamma \text{ uses }t} \pi(I)\pi(F) |\gamma_{IF}|\right\}
\end{align}
where $t$ runs over all transitions, i.e., all pairs of adjacent state of the chain; and $|\gamma_{IF}|$ denotes the length $l$ of the path $\gamma_{IF}$. In the analysis in~\cite{jerrum2003counting}, Jerrum modified the chain by making it ``lazy'', i.e., with probability $1/2$ the chain stays in the same state, and otherwise makes the transition specified in the original chain. This laziness doubles the mixing time, but ensures that the eigenvalues of the transition matrix are all non-negative, which helps the analysis. The relation between the congestion and the mixing time of a lazy MC is formalized as follows:
\begin{theorem}[\cite{jerrum2003counting}]\label{thm:canonical-mixing}
    The mixing time of the lazy MC is bounded by
$\tau_{\mix}(\varepsilon) \leq 2 \varrho\left(2 \ln \varepsilon^{-1}+\ln \pi_{\min}^{-1}\right)$,
where  $\varrho=\varrho(\Gamma)$ is the congestion defined in \eq{congestion} with respect to any set of canonical paths $\Gamma$, and $\pi_{\min}=\min_{M\in\Omega}\pi(M)$ is the minimum value of the stationary distribution.
\end{theorem}

Here we present a strengthened version of the multiple canonical path framework~\cite{sinclair1992improved,guo2017random}, which allows for multiple distinct paths between each pair $(I, F)$. Each path is associated with a selection probability $p^{(i)}_{IF}$, forming a complete probability distribution over the path set from $I$ to $F$. Intuitively, this framework can be viewed as a generalization of the original canonical path method, where for each pair $(I,F)$ we consider selecting a path from a collection with specific probabilities instead of a single canonical path.

Formally, we can define this multiple canonical path framework as follows:
\begin{definition}[Multiple canonical path framework]
Consider a Markov chain with state space $\Omega$.
    For any pair of state $(I,F)\in \Omega\times \Omega$, let $\Gamma_{IF} := \{(\gamma_{IF}^{(i)}, p_{IF}^{(i)})\mid i=1,2,\dots, l(I,F)\}$ be a collection of paths from $I$ to $F$ with size $l(I,F)$, where $\gamma_{IF}^{(i)}$ is the $i$-th path and $p_{IF}^{(i)}$ is the selection probability of the $i$-th path, satisfying the following conditions:
    \begin{itemize}
        \item Each $\gamma_{IF}^{(i)}$ is a path from $I$ to $F$ in the Markov chain, i.e., there exists a sequence of states $M_0 = I, M_1, \ldots, M_l = F$ such that $(M_j, M_{j+1})$ is a transition in the Markov chain for all $j=0,\ldots,l-1$.
        \item The selection probabilities $p_{IF}^{(i)}$ are non-negative and sum to 1: $p_{IF}^{(i)} \geq 0$ for all $i=1,\ldots,l(I,F)$, and $\sum_{i=1}^{l(I,F)} p^{(i)}_{IF} = 1$.
        
    \end{itemize}
    The length of each path $\gamma_{IF}^{(i)}$ is denoted as $|\gamma_{IF}^{(i)}|$, and for convenience we define $|\gamma_{IF}| = \max_{i=1,\ldots,l(I,F)} |\gamma_{IF}^{(i)}|$.

    The multiple canonical path framework is denoted as $\Gamma = \{\Gamma_{IF}\mid I,F\in\Omega\}$, and the congestion $\varrho$ is defined as follows:
    \begin{align}
         \varrho(\Gamma) := &\max_{t=(M,M')}\bigg\{\frac{1}{\pi(M)P(M,M')} \cdot 
        \sum_{I,F\in \Omega} \sum_{i=1}^{l(I,F)} \textbf{1}_{\{\gamma_{IF}^{(i)}\text{ uses }t\}} \cdot\pi(I)\pi(F) |\gamma_{IF}^{(i)}|\cdot p^{(i)}_{IF}\bigg\}.
    \end{align} 
\end{definition}

Similarly, \thm{canonical-mixing} remains valid for the enhanced multiple canonical path framework. In subsequent proofs in \append{dense-graph-proof}, we will employ this enhanced canonical path framework, as its multiple-path construction enables more effective utilization of symmetries between different transitions, thereby providing greater flexibility in establishing conductance bounds and mixing time estimates.

The key step in the proof in \thm{canonical-mixing} is constructing the relations between the congestion $\varrho$ as well as local and global variations:
\begin{lemma}\label{lem:canonical-mixing-variation}
    Let $\pi$ be the stationary distribution of a Markov chain with state space $\Omega$ and transition matrix $P$.
    For any function $f: \Omega\to\R$, define the variance of $f$ with respect to $\pi$ as
    \begin{align}
        \Var_\pi(f) := \sum_{M\in\Omega} \pi(M)\left(f(M) - E_\pi f\right)^2 = \frac{1}{2} \sum_{I,F\in \Omega} \pi(I)\pi(F) \left(f(I) - f(F)\right)^2
    \end{align}
    and the ``global variation'' of function $f$ with respect to the transition as
    \begin{align}
        \mathcal{E}_\pi(f,f) := \frac{1}{2} \sum_{M,M'\in \Omega} \pi(M)P(M,M')\left(f(M) - f(M')\right)^2.
    \end{align}
    Then for any canonical path framework $\Gamma$ of the MC, we have

    \begin{align}\label{eq:eq3}
        \mathcal{E}_\pi(f,f) \geq \frac{1}{\varrho(\Gamma)} \Var_\pi f.
    \end{align}
\end{lemma}

    To extend \thm{canonical-mixing} to the multiple canonical path scenario, it suffices to prove that \lem{canonical-mixing-variation} remains valid under the generalized framework: 
    \begin{proof}[Proof of \lem{canonical-mixing-variation} for multiple canonical path framework]
        \begin{align}
            2\operatorname{Var}_{\pi}f & =\sum_{I,F\in\Omega}\pi(I)\pi(F)\left(f(I)-f(F)\right)^{2} \\ &=\sum_{I,F\in\Omega}\pi(I)\pi(F)\sum_{i=1}^{l(I,F)}p^{(i)}_{IF}\cdot\left(\sum_{(M,M')\in\gamma^{(i)}_{IF}}1\cdot\left(f(M)-f(M')\right)\right)^2 \\
        &\leq\sum_{I,F\in\Omega}\pi(I)\pi(F)\sum_{i=1}^{l(I,F)}p^{(i)}_{IF}\cdot\left|\gamma^{(i)}_{IF}\right|\sum_{(M,M')\in\gamma^{(i)}_{IF}}\left(f(M)-f(M')\right)^{2} \\
         & =\sum_{M,M'\in\Omega}\sum_{I,F\in\Omega}\pi(I)\pi(F)\cdot \sum_{i=1}^{l(I,F)}\textbf{1}_{\{\gamma_{IF}^{(i)}\text{ uses }(M,M')\}}p^{(i)}_{IF}\left|\gamma^{(i)}_{IF}\right|\left(f(M)-f(M')\right)^2 \\
         & =\sum_{M,M'\in\Omega}\left(f(M)-f(M')\right)^{2}\cdot  \sum_{I,F\in\Omega}\pi(I)\pi(F) \sum_{i=1}^{l(I,F)}\textbf{1}_{\{\gamma_{IF}^{(i)}\text{ uses }(M,M')\}}p^{(i)}_{IF}\left|\gamma^{(i)}_{IF}\right| \\
         & \leq\sum_{M,M'\in\Omega}\left(f(M)-f(M')\right)^{2}\pi(M)P(M,M')\varrho \\
         & =2\varrho\mathcal{E}_{\pi}(f,f).
        \end{align}
    \end{proof}

Building upon the analytical framework in~\cite{jerrum2003counting}, we have extended \thm{canonical-mixing} to accommodate multiple canonical paths, achieving a stronger result:
\begin{theorem}\label{thm:canonical-mixing-multiple}
    For any Markov chain, the mixing time of its lazy MC is bounded by
$\tau_{\mix}(\varepsilon) \leq 2 \varrho\left(2 \ln \varepsilon^{-1}+\ln \pi_{\min}^{-1}\right)$,
where  $\varrho=\varrho(\Gamma)$ is the congestion defined in \eq{congestion} with respect to any multiple canonical path framework $\Gamma$, and $\pi_{\min}=\min_{M\in\Omega}\pi(M)$ is the minimum value of the stationary distribution.
\end{theorem}

\section{Uniform Sampling for Perfect Matching}\label{app:uniform-sampling-pm}

In this section, we briefly introduce algorithms for uniform sampling of perfect matchings. We consider graphs $G = (V, E)$, where $V$ is a vertex set of size $2n$ and $E$ is an edge set of size $m$. In \append{uniform-sampling-dense}, we focus on the case of general graphs and give a MC method corresponding to \lem{non-bigraph-uniform}; in \append{uniform-sampling-bipartite}, we focus on the case of bipartite graphs and give a MC method corresponding to \lem{bigraph-uniform}. Unless otherwise specified, $G$ may refer to a non-bipartite graph. 

\subsection{Uniform Sampling for Perfect Matchings on Dense Graphs}\label{app:uniform-sampling-dense}
We begin by presenting a Markov Chain Monte Carlo (MCMC) method, as proposed by~\cite{jerrum1989approximating}, for uniform sampling of perfect matchings in general dense graphs. This method provides a fully polynomial randomized approximation scheme (FPRAS) to approximate the number of perfect matchings in a balanced bipartite graph, where the method is based on sampling perfect matchings uniformly at random. The sampling method also works for general non-bipartite graphs. For a general graph $G$, let $M_k(G)$ denote the set of matchings of size $k$ in $G$. Note that $M_n(G)$ is the set of perfect matchings in $G$, and $M_{n-1}(G)$ consists of matchings that have exactly one edge fewer than a perfect matching. For convenience, we refer to these as ``near-perfect'' matchings in $G$.
Let $\mathcal{N}$ be the set $M_n(G)\cup M_{n-1}(G)$.

Jerrum proposed a Markov chain to sample a matching $M$ from $\mathcal{N}$ uniformly at random, which is specified as follows: In any state $M\in\mathcal{N}$, we proceed by selecting an edge $e=(u,v)\in E$ uniformly at random. The transition is determined as follows:
\begin{itemize}
    \item If $M\in M_n(G)$ and $e\in M$, then move to state $M' = M \backslash \{e\}$;
    \item If $M\in M_{n-1}(G)$ and $u,v$ are unmatched in $M$, then move to state $M' = M\cup \{e\}$;
    \item If $M\in M_{n-1}(G)$, $u$ is matched to $w$ in $M$ and $v$ is unmatched in $M$, then move to state $M' = M \cup \{e\} \backslash \{(w,u)\}$; symmetrically, if $v$ is matched to $w$ and $u$ is unmatched, then move to state $M' = M \cup \{e\} \backslash \{(w,v)\}$;
    \item In all other cases, do nothing.
\end{itemize}

Note that the Markov chain is ergodic, aperiodic, and irreducible, meaning that it converges to a unique stationary distribution~\cite{levin2017markov}. 
For two near-perfect matchings $M$ and $M'$ such that $M'= M\cup \{(u,v)\} \backslash \{(w,u)\}$, the corresponding transition probability is formulated as
\begin{align}
    \Pr[M\to M'] = \frac{1}{|E|};\quad \Pr[M'\to M] = \frac{1}{|E|}.
\end{align}
For a perfect matching $M$ and a near-perfect matching $M'$ such that $M'= M\backslash \{(u,v)\}$, the corresponding transition probability is formulated as
\begin{align}
    \Pr[M\to M'] = \frac{1}{|E|};\quad \Pr[M'\to M] = \frac{1}{|E|}.
\end{align}
Therefore the stationary distribution of the Markov chain is given by $\pi(M) = \frac{1}{|\mathcal{N}|}.$

To bound the mixing time of this Markov chain, we introduce the concept of the conductance of a Markov chain, which is also known as bottleneck ratio: 
\begin{align}
    \Phi := \min_{S:0< \sum_{i\in S} \pi_i\leq \frac{1}{2}} \frac{\sum_{i\in S, j\not\in S}p_{ij}\pi_i}{\sum_{i\in S}\pi_i}\label{eq:conductance}
\end{align}
where $p_{ij}$ is the transition probability from state $i$ to state $j$, and $\pi_i$ is the stationary probability of state $i$.

By utilizing the relationship between the conductance $\Phi$ and the second eigenvalue of the transition matrix, we can bound the mixing time $t_{\mix}(\eps)$ of the Markov chain as follows:
\begin{lemma}[\cite{jerrum1989approximating,sinclair1989approximate}]
    If the given Markov chain is irreducible, reversible, and ``lazy'', i.e., $\min_{i} p_{ii}\geq \frac{1}{2}$, we have the following relation between the conductance $\Phi$ and the mixing time $t_{\mix}(\eps)$:
    \begin{align}
        t_{\mix}(\eps) \leq \frac{2}{\Phi^2} \left(\ln {\pi_{\operatorname{min}}}^{-1}+\ln \eps^{-1}\right).
    \end{align}
\end{lemma}
Jerrum provided a polynomial bound of the conductance $\Phi$ if the graph $G$ is dense:
\begin{theorem}[\cite{jerrum1989approximating}]
    If the minimum degree of the graph $G$ is at least $n$, then the conductance $\Phi$ of the Markov chain is lower bounded by $\Phi\geq \frac{1}{64}n^6$. If $G$ is also a balanced bipartite graph, then the conductance $\Phi$ of the Markov chain is bounded $\Phi\geq \frac{1}{12}n^6$. 
\end{theorem}
From the above conductance bound, we obtain a polynomial bound on the mixing time of the Markov chain.
\begin{corollary}[\cite{jerrum1989approximating}]
    If the minimum degree of the graph $G$ is at least $n$, then the mixing time $t_{\mix}(\eps)$ of the Markov chain is bounded by
    \begin{align}
        t_{\mix}(\eps) = \tilde{O}(n^{12}\ln \eps^{-1}).
    \end{align}
\end{corollary}

To sample a perfect matching uniformly at random, we can just repeatedly sample a matching $M\in\mathcal{N}$ from the Markov chain until we get a perfect matching $M\in M_n(G)$. After $t$ steps, the failure probability of getting a perfect matching is $(1-p)^t$, where $p$ is the probability of getting a perfect matching from the Markov chain. Then we can give a polynomial bound of $p$:

\begin{lemma}\label{lem:near-perfect-matching}
    If the minimum degree of the graph $G$ is at least $n$, then we have
    \begin{align}
        \frac{|M_n(G)|}{|M_{n-1}(G)|} \geq \frac{1}{2n^2}.
    \end{align}
\end{lemma}
\begin{proof}
First, we consider mapping each near-perfect matching to a perfect matching.
For any near-perfect matching $M\in M_{n-1}(G)$, suppose that $u,v$ are the unmatched vertices in $M$. We construct the mapping by considering two cases:

    Case 1: If $(u,v)\in E$, then we map $M$ to a perfect matching $M\cup\{u,v\}$.

    Case 2: If $(u,v)\not\in E$, we claim that there exists $(w,z)\in M$ such that $(u,w),(v,z)\in M$. Otherwise, since $\deg(u)\geq n$ and $\deg(v)\geq n$, let $N(u):= \{w|w \text{ is the neighbor of }u\}$, $N'(u):= \{w|(w,z)\in M \text{ and }z \text{ is the neighbor of }u\}$ and $N(v):= \{w|w \text{ is the neighbor of }v\}$. Thus we have $|N'(u)| = |N(u)|\geq n$, $|N(v)|\geq n$ and $N'(u) \cap N(v) = \emptyset$, but $N'(u) \cup N(v)\subseteq V\backslash \{u,v\}$, which is a contradiction. Therefore, we can map $M$ to a perfect matching $M\cup\{(u,w),(v,z)\}\backslash\{(w,z)\}$.
    
    For any perfect matching $M\in M_n(G)$, now we compute the number of near-perfect matchings $M'\in M_{n-1}(G)$ that can be mapped to $M$ by the above scheme. For any edge $(u,v)\in M$, we can find at most one near-perfect matching $M'\in M_{n-1}(G)$ such that $M' = M\backslash \{(u,v)\}$ can be mapped to $M$ by case 1 of the above scheme. Thus at most $n$ near-perfect matchings $M'\in M_{n-1}(G)$ can be mapped to $M$ by case 1 of the above scheme. Any near-perfect matching $M'\in M_{n-1}(G)$ that can be mapped to $M$ by case 2 must be of form $M\cup\{(u,v)\}\backslash\{(u,w),(v,z)\}$. We have $\binom{n}{2}$ choices of $\{(u,w),(v,z)\}$, and each choice gives at most $4$ possible near-perfect matchings. 
    Therefore, we can find at most $n+4\cdot \binom{n}{2} = 2n^2-n$ near-perfect matchings $M'\in M_{n-1}(G)$ such that $M'$ can be mapped to $M$ by the above scheme, and we hence have
    \begin{align}
        \frac{|M_n(G)|}{|M_{n-1}(G)|} \geq \frac{1}{2n^2-n} 
        &\geq \frac{1}{2n^2}. 
    \end{align}
\end{proof}

Finally, we achieve a uniform sampling of perfect matchings in $G$ in polynomial time, as shown in \lem{non-bigraph-uniform}.
\begin{proof}[Proof of \lem{non-bigraph-uniform}]
    For the given failure probability $\eta$, we can run the Markov chain to uniformly sample a matching $M\in\mathcal{N}$ at random with error $\frac{\eta}{2t}$ for $t = \lceil(2+4n^2)\ln {(\frac{\eta}{2})}^{-1}\rceil$ times, performing $t_{\mix}\left(\frac{\eta}{2t}\right)$ steps each time, and take the first sampled perfect matching as output.
    
    By \lem{near-perfect-matching}, the probability of getting a uniformly random perfect matching from the Markov chain is 
    $p \geq \frac{1}{1+2n^2} \cdot (1-\frac{\eta}{2t})  \geq \frac{1}{1+2n^2} \cdot \frac{1}{2} \geq \frac{1}{2+4n^2}$, and thus the failure probability of finally getting a uniformly random perfect matching is bounded by
    \begin{align*}
        \Pr[\text{failure}] &\leq (1-p)^t +\frac{\eta}{2t} \cdot t \leq (\frac{1}{e})^{\ln{(\frac{\eta}{2})^{-1}}} +\frac{\eta}{2} = \eta.
    \end{align*}
    and the total time is bounded by
    \begin{align*}
       t\cdot t_{\mix}\left(\frac{\eta}{2t}\right) = \tilde{O}(n^{2}\ln {\eta^{-1}})\cdot \tilde{O}(n^{12}\ln {\eta^{-1}}) = \tilde{O}(n^{14}(\ln {\eta^{-1}})^2).
    \end{align*}
\end{proof}

\subsection{Uniform Sampling of Perfect Matchings on Bipartite Graphs}
\label{app:uniform-sampling-bipartite}

The aforementioned algorithm fails when the number of near-perfect matchings significantly exceeds that of perfect matchings. Building on this foundation, Jerrum proposed a novel algorithm in~\cite{jerrum2004polynomial} that achieves uniform sampling of perfect matchings for general balanced bipartite graphs in $O(n^{11}(\log n)^2(\log n + \log \eta^{-1}))$ time. The formal statement is shown in \lem{bigraph-uniform}.

In a nutshell, \cite{jerrum2004polynomial} introduced an additional weight factor that takes account of the holes in near-perfect matchings to solve this problem. For the given bipartite graph $G=(V_1,V_2,E)$ with $|V_1|=|V_2| = n$, the set $\mathcal{N}$ is separated into $n^2 + 1$ patterns, including one pattern $\M$ of all perfect matchings and $n^2$ patterns $\M(u,v)$ of other near perfect matchings. Here set $\M(u,v)$ consists of near-perfect matchings that have holes at the vertices $v\in V_1$ and $u\in V_2$, where these vertices remain unmatched in the respective matching.
Our goal is to define the weights of $n^2+1$ patterns properly so that the sampling probabilities for each pattern are identical, ensuring that the total probability mass for perfect matchings is $\Omega(\frac{1}{n^2})$. This guarantees that perfect matchings can be sampled in polynomial time with high probability.

For each pair of vertices $(u,v)\in E$, we introduce a positive activity $\lambda(u,v)$. This concept is extended to matchings $M$ of any size by $\lambda(M) = \prod_{(u,v)\in M}\lambda(u,v)$. For a set of matchings $\mathcal{S}$ we define $\lambda(\mathcal{S}) = \sum_{M\in \mathcal{S}} \lambda(M)$. This edge weight works with a complete bipartite graph on $2n$ vertices, and for any arbitrary bipartite graph we can set $\lambda(e) = 1$ for $e\in E$ and $\lambda(e) = \alpha \approx 0$ for $e\not\in E$. Taking $\alpha \leq \frac{1}{n!}$ ensures that the “bogus” matchings have little effect.

Now we specify the desired stationary distribution of the Markov chain. This will be the distribution $\pi$ over $\mathcal{N}$ defined by $\pi(M) \propto\Lambda(M)$, where
\begin{align}
\label{eq:bipartite-stationary-probabilities}
    \Lambda(M) = \begin{cases}
        \lambda(M) w(u,v) &\text{ if } M\in \M(u,v) \text{ for some } u,v;\\
        \lambda(M) &\text{ if } M\in \M,
    \end{cases}
\end{align}
here $w\colon V_1 \times V_2 \to \R^+$ is weight function for the hole of nearly perfect matching to be specified later. The Markov chain is a slight variant of the original chain by~\cite{jerrum1989approximating}. In any state $M\in\mathcal{N}$:
    \begin{itemize}
    \item If $M\in \M$, choose an edge $e= (u,v)$ uniformly at random from $M$; then move to $M' = M\backslash\{e\}$.
    \item If $M\in \M(u,v)$, choose $z$ uniformly at random from $V_1 \cup V_2$:
    \begin{itemize}
        \item if $z\in \{u,v\}$ and $(u,v)\in E$, move to $M' = M + (u,v)$;
        \item if $z\in V_2, (u,z)\in E$ and $(x,z)\in M$, move to $M' = M+(u,z) - (x,z)$;
        \item if $z\in V_1, (z,v)\in E$ and $(z,y)\in M$, move to $M' = M+(z,v) - (z,y)$;
        \item otherwise move to $M' = M$.
    \end{itemize}
    \item With probability $\min\{1,\Lambda(M') / \Lambda(M)\}$ go to $M'$, otherwise stay at $M$.
\end{itemize}

Thus the non-trivial transitions are of three types: removing an edge from a perfect matching; adding an edge to a near perfect matching; and exchanging an edge of a near-perfect matching with another edge adjacent to one of its holes. It is easy to verify that the stationary probabilities satisfy \eq{bipartite-stationary-probabilities}.

Next we need to specify the weight function $w$. Ideally we would like to take $w=w^*$, where
\begin{align}
    w^*(u,v) = \frac{\lambda(\M)}{\lambda(\M(u,v))}
\end{align}
for each pair $(u,v)$ with $\M(u,v)\neq \emptyset$. With this choice of weights, we have $\Lambda(\M(u,v)) = \Lambda(\M)$, i.e., each pattern is equally likely under the stationary distribution. 

However, precisely calculating the ratio of $\lambda$ is as difficult as computing the Hafnian. Here, instead of directly calculating $w^*$ exactly, we find an approximation that satisfies:
\begin{align}
\label{eq:w-approximation}
    \frac{w^*(u,v)}{2} \leq w(u,v) \leq 2w^*(u,v),
\end{align}
with very high probability. This perturbation will reduce the relative weight of perfect matchings by at most a constant factor. 

For $w$ that meets the above requirements, Ref.~\cite{jerrum2004polynomial} proved the following mixing time:

\begin{theorem}
\label{thm:bipartite-mixing-time}
    Assuming the weight function w satisfies \eq{w-approximation} for all $(u,v)\in V_1 \times V_2$ with $\M(u,v) \neq \emptyset$, then the mixing time of the above Markov chain is bounded above by
    \begin{align}
        O(n^6 (\log(\pi(M)^{-1} +\log \eps^{-1}))).
    \end{align}
\end{theorem}

Finally, we need to address the issue of computing (approximations to) the weights $w^*$ of the given graph $G$. 
We can construct a special edge activity scheme such that $\lambda_{K_{n,n}}(\M)\approx \lambda_G(\M)$ and $\lambda_{K_{n,n}}(\M(u,v))\approx \lambda(\M(u,v))$, where $K_{n,n}$ is the complete balanced bipartite graph on $2n$ vertices.
Consequently, it suffices to estimate $\frac{\lambda_{K_{n,n}}(M)}{\lambda_{K_{n,n}}(M(u,v))}$. Jerrum provided an iterative method: 
we initialize the edge activities to trivial values, for which the corresponding $w^* = \frac{\lambda_{K_{n,n}}(\M)}{\lambda_{K_{n,n}}(\M(u,v))}$ can be computed easily. Then, we gradually adjust the edge activities towards the actual edge activities in graph $G$. During each adjustment of edge activity, we use the estimate of $w^*$ from the previous step to update our estimate of $w^*$.

\begin{itemize}
    \item First, we determine the target value of edge activities $\lambda_G$ in the graph $G$. For the desired uniform sampling of perfect matchings of graph $G$, {we can assign values $\lambda_G(e) = 1$ for all $e\in E$ and $\lambda_G(e) = \frac{1}{n!}$ for all $e\not\in E$.} Since the activity of an invalid matching of $G$ (i.e., one that includes a non-edge with activity $\frac{1}{n!}$) is at most $\frac{1}{n!}$ and there are at most $n!$ possible matchings (and thus at most $n!$ invalid matchings), the combined activity of all invalid matchings is at most $1$. Assuming the graph has at least one perfect matching, the invalid matchings merely increase by at most a small constant factor the probability that a single simulation fails to return a perfect matching. More specifically,
    \begin{align}
        \lambda_{K_{n,n}}(\M)-1\leq &\lambda_G(\M) \leq \lambda_{K_{n,n}}(\M); \\
        \lambda_{K_{n,n}}(\M(u,v)) - 1 \leq &\lambda_G(\M(u,v)) \leq \lambda_{K_{n,n}}(\M(u,v)) \text{ for }\M(u,v)\neq \emptyset\text{ in } G.
    \end{align} 
    \item Next, we choose to initialize the edge activity as follows:
    \begin{align}
        \lambda(u, v) = 1 \quad \text{for all pairs } (u,v).
    \end{align}
    It is straightforward to compute that $\lambda_{K_{n,n}}(\M) = n! \quad \text{and} \quad \lambda_{K_{n,n}}(\M(u,v)) = (n - 1)!$, thus the ideal $w^*$ of current activity is
    \begin{align}
        w^*(u,v) = \frac{\lambda_{K_{n,n}}(\M)}{\lambda_{K_{n,n}}(\M(u,v))} = n.
    \end{align}
    \item In each iteration, an appropriate vertex $v$ is selected, and for each pair $(u,v)\not\in E$, we update its activity $\lambda(u,v)$ to $\lambda'(u,v)$, where $\lambda(u,v) e^{-1/2} \leq \lambda'(u,v) \leq \lambda(u,v) e^{1/2}$. We continue updating until $\lambda_G(u,v) = \frac{1}{n!}$ for each pair $(u,v)\not\in E$.
    
    During one iteration, we give an approximation to current $w^*$ within ratio $c$ for some $1<c<2$ for the current $\lambda$.

    We may use the identity
    \begin{align}
        \frac{w(u,v)}{w^*(u,v)} = \frac{\pi(\M(u,v))}{\pi(\M)},
    \end{align}
    since $w(u,v)$ is known to us, and the probabilities on the right side may be estimated to arbitrary precision by taking several samples and take the average. \thm{bipartite-mixing-time} allows us to sample, in polynomial time, from a distribution $\hat{\pi}$ that is within variation distance $\delta$ of $\pi$. 

    The algorithm to estimate the current ideal weights $w^*(u,v)$ is presented in \algo{approximating-ideal-weight}.

\begin{algorithm}[h]
        \KwIn{input parameters $c=\frac{6}{5}$, $G=(V_1,V_2,E)$, $\eta$(failure probability)}
        \KwOut{an approximation of the ideal weights $w^*(u,v)$ of graph $G$}
            Set $\hat{\eta}\leftarrow O(\eta/(n^4\log n))$(failure probability in each MC), $S\leftarrow O(n^2\log (1/\hat{\eta}))$ (number of samples in each round), $T\leftarrow O(n^7 \log n)$ (running time of the Markov chain);\;
            
            Initialize $\lambda(u,v)\leftarrow 1$ for all pairs $(u,v)\in V_1\times V_2$;\;
            
            Initialize $w(u,v)\leftarrow n$ for all pairs $(u,v)\in V_1\times V_2$;\;
        
            \lWhile{there exists a pair $y,z$ with $\lambda(y,z)>\lambda_G(y,z)$}\
            {
                \Indp Take a sample of size $S$ from the Markov chain with function $\lambda,w$, using a simulation of $T$ steps in each case. In each simulation let $\alpha_i(u,v)$ denote the proportion of samples with hole $(u,v)$, and $\alpha_i$ denote the proportion of samples that are perfect matchings;\;
        
                For the ideal weight $w^*$ of the graph in the current round, se the sample to obtain estimates $w'(u,v):= {w(u,v)\cdot  \frac{\sum_{i=1}^S \alpha_i}{\sum_{i=1}^S  \alpha_i(u,v)}}$ satisfying $w^*(u,v)/c \leq w'(u,v) \leq w^*(u,v)*c$, for all $u,v$ with high probability; \;

                Set $\lambda(y,v) \leftarrow \max\{\lambda(y,v) e^{-1/2},\lambda_G(y,v)\}$ for all $v\in V_2$, and $w(u,v)\leftarrow w'(u,v)$ for all $u,v$;
                
            }
        Output the final weight $w(u,v)$.
        \caption{The algorithm for approximating the ideal weights.}
        \label{algo:approximating-ideal-weight}
    \end{algorithm}
\end{itemize}



\section{Polynomial Mixing Time for Dense Graphs}\label{app:dense-graph-proof}
This section presents formal proofs of our main theoretical contributions: polynomial mixing time guarantees for \algo{double-loop} on dense bipartite graphs and non-bipartite graphs. Specifically, We start with the cases of complete bipartite graphs and complete graphs, followed by dense graphs where the minimum degree differs from that of a complete graph by only a constant.

\append{complete-bipartite} proves polynomial mixing time on complete bipartite graphs. \append{complete-non-bipartite} proves polynomial mixing time on complete graphs.
\append{dense-bipartite} and \append{dense-non-bipartite} extend these results to dense graphs, completing the proofs of \thm{main-results-bigraph} and \thm{main-results-non-bigraph}.
In \append{non-dense-graph} we discuss the technical barriers for generating our theoretical results to non-dense graphs.

First, we provide the definitions of the canonical paths and the congestion for the matching MC. Given two matchings $I$ (initial matching) and $F$ (final matching), we need to connect $I$ and $F$ by several canonical paths $\gamma_{IF}^{(i)}$ in the adjacency graph of the matching MC, with each path $\gamma_{IF}^{(i)}$ having a corresponding probability distribution $p^{(i)}(I,F)$, which is the probability of choosing the path $\gamma_{IF}^{(i)}$ when we need to ``travel'' from $I$ to $F$ in the matching MC. 
Denote the set of canonical paths that involve the transition $t$ by $\mathrm{cp}(t):= \{(I,F,i)\mid t\in \gamma^{(i)}_{IF}\}$. Following from \thm{canonical-mixing-multiple}, we need to bound the congestion
\begin{align}\label{eq:congestion-matching}
     \varrho := &\max_{t=(M,M')}\bigg\{\frac{1}{\pi(M)P(M,M')}
    \cdot\sum_{(I,F,i)\in \mathrm{cp}(t)} p^{(i)}(I,F)\cdot \pi(I)\pi(F) |\gamma_{IF}^{(i)}|\bigg\}.
\end{align}
The specific canonical paths between $I$ and $F$, along with their corresponding probability distributions, will be formally defined in subsequent sections.


\subsection{Complete Bipartite Graph $G_{m,n}$ } \label{app:complete-bipartite}
In this subsection, we provide a proof of polynomial mixing time for \algo{double-loop} on the complete bipartite graph $G_{m,n}$. For convenience, suppose $m\geq n$.  
Unlike methods in~\cite{jerrum2003counting}, for the case of a complete graph, we can directly compute the value of the expression in \eq{congestion-matching}. 

Before delving into detailed analysis, we outline the proof strategy as follows. The key insight lies in leveraging the inherent symmetry of the complete graph. This symmetry allows us to establish a multiple canonical path scheme for complete bipartite graphs that symmetric transitions share identical congestion values. (It also works for complete graphs, which is shown in \append{complete-non-bipartite}.) This symmetry property enables us to compute the total congestion across these transitions and subsequently derive the congestion bound for one transition through direct computation.

\paragraph{Multiple Canonical Path Construction}

Now we proceed by formally defining the canonical paths connecting each pair of matchings $I$ and $F$. We first present a construction with strong symmetry, and then modify it to simplify our calculations. The construction is as follows:
\begin{itemize}
    \item Inspired by the original construction in~\cite{jerrum2003counting}, we first define a valid canonical path connecting matching $I$ and $F$. Notice that along any canonical path, we will have to lose or gain at least the edges in the symmetric difference $I \oplus F$; these edges define a graph of maximum degree two, which decomposes into a collection of paths and even-length cycles, each of them alternating between edges in $I$ and edges in $F$. To give a path from $I$ to $F$, we can consider all the permutations of these connected components of $(V,I \oplus F)$. To get from $I$ to $F$, for one permutation denoted sequentially by $P_1,P_2,\dots, P_r$, we then consider to process the components of $(V,I\oplus F)$ in the order $P_1,P_2,\dots, P_r$. 
    \item For each connected component $P_i$, we assign an order to all its edges by selecting a ``start vertex'' and then sequentially choosing each edge in a specified direction. If $P_i$ is a path, the start vertex can be chosen as one of the endpoints of the path, and then all edges are numbered sequentially from this start vertex towards the other end. If $P_i$ is a cycle, the starting vertex can be any point $u$ on the cycle; then considering the edge $e=(u,v)$ in matching $I$ connected to this point, all edges are numbered sequentially in the direction from $u$ to $v$ on edge $e$ (which could be either clockwise or counterclockwise).
    
    We then construct the canonical path on $P_i$ using the above numbering as follows: If the edge with the smallest number is in the matching $I$, 
    we alternately perform the following operations: remove the unprocessed edge with the smallest number in $I\cap P_i$, and then add the unprocessed edge with the smallest number in $F\cap P_i$. Finally, add the last edge in $F$ (if any). If the edge with the smallest number is in the matching $F$, first add that edge, and then proceed with the aforementioned operations.

    \item We perform the above operations for all components $P_i$ in the order $P_1,\dots, P_r$, and finally give a path from $I$ to $F$. Since $I \oplus F$ contains at most $2n$ edges, the length of the path naturally does not exceed $2n$.
    
    \item Finally, we define the probability distribution of all possible canonical paths. To preserve symmetry among different paths, we assign the same probability to each canonical path from $I$ to $F$. Intuitively, we can interpret this as at each step of the process (including choosing the permutation of connected components and selecting the start vertex within each component), we select one of the possible choices uniformly at random.
\end{itemize}

For convenience, let $T^+(k):= \{(M,M+e) \mid \text{$M$ is a matching of size $k$}\}$ and $T^-(k):= \{(M,M-e) \mid$\\
$\text{$M$ is a matching of size $k$}\}$. The critical observation is that for complete bipartite graphs, our construction ensures that in one set $T^+(k)$ (or $T^-(k)$), each transition shares identical congestion values. Notice that the original construction in~\cite{jerrum2003counting} only considers ordering the vertices in $V$, and obtains a unique ordering of components by the smallest vertex. The start vertex is defined as the smallest vertex in a cycle, and the smaller of the two endpoints in a path. This construction is not sufficient to ensure the identical congestion requirement. 

Finally, we adjust the first step of our construction to facilitate subsequent analysis. Instead of considering all possible permutations, we will now limit our consideration to a specific subset.

\begin{itemize}
    \item For the connected components of $(V, I \oplus F )$, denoted by $\{P_1,P_2,\dots, P_r\}$, we classify these connected components into three sets based on the change in the number of edges before and after the transformation: Let $A_0$ denote the set where the number of edges remains unchanged, $A_1$ the set where the number of edges increases by one, and $A_{-1}$ the set where the number of edges decreases by one. We only consider permutations satisfying the following requirement: Components in $A_0$ are traversed first, followed by components from $A_1$ and $A_{-1}$ alternating until one of the sets is exhausted; finally, the remaining components are traversed. 
    The number of these limited permutations is $|A_0|! \cdot |A_1|! \cdot |A_{-1}|!$.
    For each permutation scheme, we sequentially transform each connected component from its initial configuration in $I$ to its corresponding final configuration in $F$. 
    
    \item For each connected component, in the case of a cycle $P_i$, we consider all cases where every vertex on the cycle becomes the ``start vertex'', which has $|P_i|$ choices, and then still handle the edges one by one; in the case of a path $P_i$, we consider the two cases where the two endpoints become the ``start vertex'', and matches are then processed one by one. 
    \item All possible canonical paths from $I$ to $F$ take the same probability.
\end{itemize}
This final construction still satisfies the requirement of identical congestion for all transitions in $T^+(k)$ (or $T^-(k)$):
\begin{proposition}\label{prop:congestion-symmetry-bipartite}
    For complete bipartite graphs, under the aforementioned multiple canonical path construction, the congestion of each transition in $T^+(k)$ (or $T^-(k)$) is identical.
\end{proposition}

\begin{proof}[Proof of \prop{congestion-symmetry-bipartite}]
    For two distinct transitions $t$ and $t'$ in a same $T^+(k)$ (or $T^-(k)$), there exists a vertex transposition $(a, b)$ \textemdash a bijective mapping on the vertex set $V$, that transforms $t$ into $t'$. Then, for any two matchings $I$ and $F$, applying the same vertex transposition to both matchings $I$ and $F$ yields matchings $I'$ and $F'$, respectively. Under this transformation, the congestion of transition $t$ caused by the canonical paths between $I$ and $F$ equals precisely the congestion of transition $t'$ caused by the corresponding paths between $I'$ and $F'$, demonstrating the symmetry of the congestion under vertex permutations.    
\end{proof}

The improved construction has a key property: the size of every matching along the path from $I$ to $F$ is bounded between the sizes of $I$ and $F$. In other words, a path between matchings of small size does not pass through a matching of large size, and conversely, a path between matchings of large size does not pass through a matching of small size.
\begin{proposition}\label{prop:canonical-path-size}
    For any graph $G$ and any pair of matchings $(I, F)$ in $G$, all canonical paths $\gamma_{I,F}^{(i)}$ selected under our rules between $I$ and $F$ satisfy: 
    \begin{align}
       \min\{|I|,|F|\} -1\leq |\gamma_{I,F}^{(i)}|\leq \max\{|I|,|F|\} + 1.
    \end{align}
\end{proposition}

\sfig{canonical-path} shows an example of one canonical path from $I$ to $F$.
\begin{figure*}[!htbp]
    \centering
        \includegraphics[width=\linewidth]{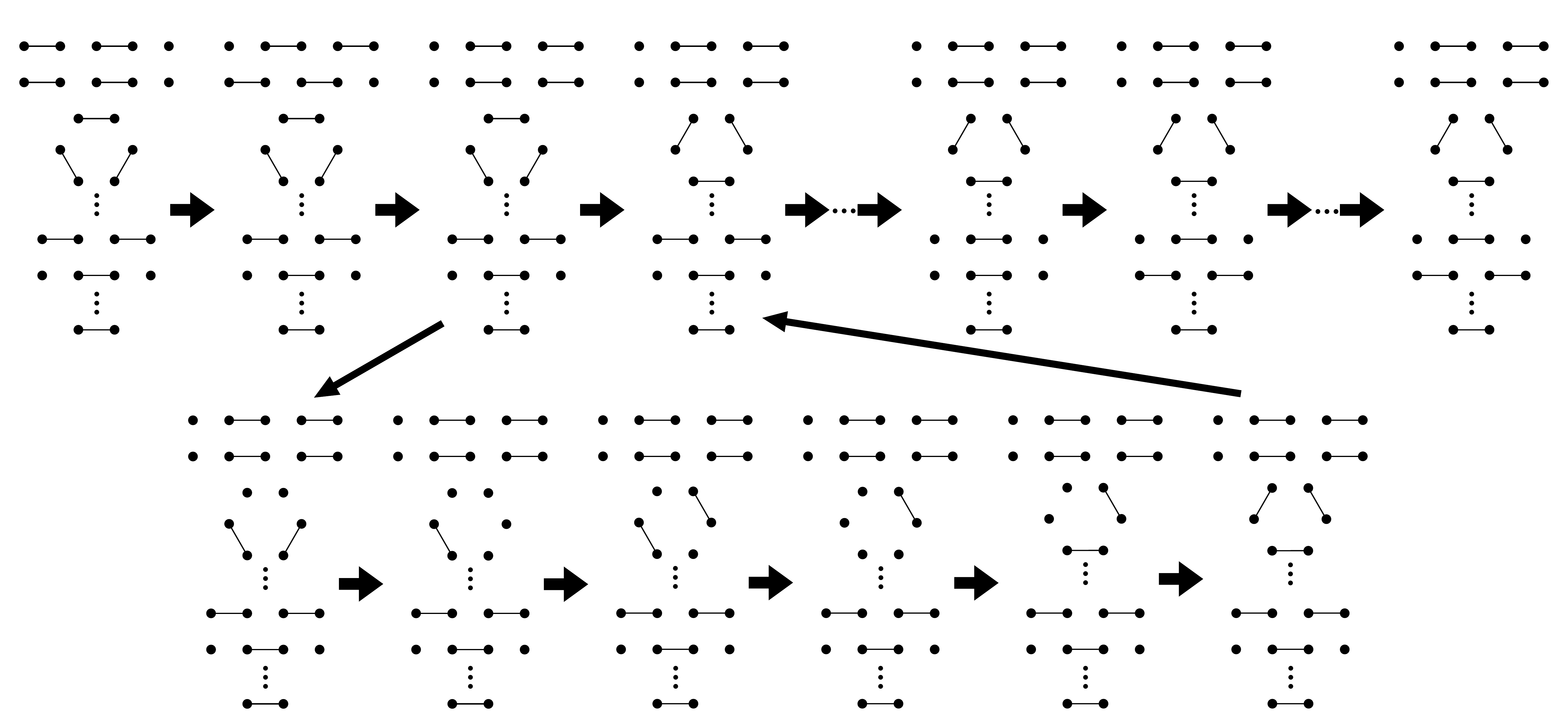}
        \caption{A possible canonical path between matchings $I$ and $F$.}
        \label{sfig:canonical-path}
\end{figure*}

\paragraph{Congestion of Multiple Canonical Path}

Now we can estimate the congestion of one transition $t$ in the complete bipartite graph $G_{m,n}$. 

The first step is to calculate $\pi(M)$ for each matching. Since any matching $M$ expands a complete bipartite subgraph of size $2|M|$, we have $\Haf(M) = |M|!$, and thus $\pi(M)\propto \lambda^{2|M|} |M|!$.
Then for edge sets of any two matchings $M$ and $M \cup\{e\}$, we have
\begin{align}
    \Pr[M\to M\cup \{e\}] &=  \frac{1}{mn}\cdot \frac{\lambda^2}{1+\lambda^2},\\
    \Pr[M \cup \{e\}\to M] &= \frac{1}{mn}\cdot \frac{\Haf(G_M)}{\Haf(G_{M\cup\{e\}})}\cdot \frac{1}{1+\lambda^2} \geq \frac{1}{mn}\cdot \frac{1}{|M|+1}\cdot \frac{1}{1+\lambda^2}.
\end{align}
Now consider the congestion of a transition $t=(M,M')$ in $T^+(k)$ (or $T^-(k)$):
\begin{align}
\varrho_t := \frac{1}{\pi(M)P(M,M')} \sum_{(I,F,i)\in \mathrm{cp}(t)} p^{(i)}(I,F) \pi(I)\pi(F) |\gamma_{IF}^{(i)}|.
\end{align}
We can compute the total congestion caused by all transitions in $T^+(k)$ (or $T^-(k)$).

Let $N_p$ denote the number of matchings of size $p$, and $N_{p,q}$ denote the number of transitions from a matching of size $p$ to a matching of size $q$ with $q=p\pm 1$. It can be easily calculated that 
\begin{align}  
N_p = \binom{n}{p}\binom{m}{p}\cdot p!,\quad N_{p,q} = \begin{cases}
    N_p\cdot (n-p)(m-p) &\text{ if } q = p + 1;\\
    N_p\cdot p &\text{ if } q = p - 1.
\end{cases}
\end{align}

We first consider $t\in T^+(k)$. Suppose $t=(M,M')$ with $|M| = k$ and $|M'|=k+1$, note that by \prop{canonical-path-size}, if the size of $I$ and the size of $F$ are either both less than $k - 1$ or both greater than $k + 1$, then any path from $I$ to $F$ cannot pass through any transition in the set $T^+(k)$. Therefore, we have 
\begin{align}
    \quad N_{k,k+1}\cdot \varrho_t 
    =&  \frac{1}{\pi(M)P(M,M')} \sum_{t'=(M',M'+e'):|M| = |M'|}  \sum_{(I,F,i)\in \mathrm{cp}(t')}p^{(i)}(I,F) \pi(I)\pi(F) |\gamma_{IF}^{(i)}|\\
    \leq&  \frac{1}{\pi(M)P(M,M')} \sum_{t'=(M',M'+e'):|M| = |M'|} \bigg(\sum_{|I|\geq k-1,|F|\leq k+1}\sum_{i:(I,F,i)\in \mathrm{cp}(t')}p^{(i)}(I,F) \pi(I)\pi(F) |\gamma_{IF}^{(i)}|\notag\\
    & +\sum_{|I|\leq k+1,|F|\geq k-1}\sum_{i:(I,F,i)\in \mathrm{cp}(t')}p^{(i)}(I,F) \pi(I)\pi(F) |\gamma_{IF}^{(i)}|\bigg)\\
    =& \frac{1}{\pi(M)P(M,M')}\bigg(\sum_{|I|\geq k-1,|F|\leq k+1}|\gamma_{IF}^{(i)}|\sum_{i=1}^{l(I,F)}p^{(i)}(I,F) \pi(I)\pi(F) |\gamma_{IF}^{(i)}|\notag\\
    &+\sum_{|I|\leq k+1,|F|\geq k-1}|\gamma_{IF}^{(i)}|\sum_{i=1}^{l(I,F)}p^{(i)}(I,F) \pi(I)\pi(F) |\gamma_{IF}^{(i)}|\bigg)\\
    \leq&  \frac{1}{\pi(M)P(M,M')} \bigg(\sum_{I:|I|\leq k+1} \sum_{F:|F|\geq k-1}\pi(I)\pi(F) |\gamma_{IF}^{(i)}|^2 + \sum_{I:|I|\geq k-1} \sum_{F:|F|\leq k+1}\pi(I)\pi(F) |\gamma_{IF}^{(i)}|^2\bigg)\\
    =& 2\cdot  \frac{1}{\pi(M)P(M,M')} \sum_{I:|I|\leq k+1} \sum_{F:|F|\geq k-1}\pi(I)\pi(F) |\gamma_{IF}^{(i)}|^2 \\
    \leq& 2\cdot \frac{mn(1+\lambda^2)}{\lambda^2} \cdot (2n)^2\cdot N_{k,k+1} \cdot \frac{1}{N_{k,k+1}} \cdot \frac{\sum_{l=1}^{k+1} N_l\cdot \lambda^{2l} l! \cdot \sum_{l=k-1}^n N_l\cdot \lambda^{2l} l!}{\lambda^{2k}k!\cdot \sum_{l=1}^n N_l\cdot \lambda^{2l} l!}.
\end{align}

Notice that $\frac{N_{l+1}\cdot \lambda^{2(l+1)}(l+1)!}{N_l\cdot \lambda^{2l}l!} = \lambda^2 (n-l)(m-l)\geq 4\lambda^2$ for $l<n-1$, when $\lambda>\frac{1}{2}$, we have
\begin{align}
     \frac{\sum_{l=1}^{k+1} N_l\cdot \lambda^{2l}l!}{N_{k,k+1} \lambda^{2k}k!}
    =&\frac{1}{(n-k)(m-k)} \cdot \frac{\sum_{l=1}^{k+1} N_l\cdot \lambda^{2l}l!}{N_{k+1}\cdot \lambda^{2(k+1)}(k+1)!}\cdot \frac{N_{k+1}\cdot \lambda^2(k+1)}{N_k}\\
    \leq& \frac{1}{(n-k)(m-k)}\cdot \left(\frac{4\lambda^2}{4\lambda^2-1}+1\right)\cdot \lambda^2 (n-k)(m-k)\\=&\lambda^2\left(\frac{4\lambda^2}{4\lambda^2-1}+1\right).
\end{align}

Then we consider $t\in T^-(k)$. Suppose $t=(M,M')$ with $|M| = k$ and $|M'|=k-1$, through similar analysis we have
\begin{align}
 N_{k,k-1}\cdot \varrho_t 
 \leq &2\cdot  \frac{1}{\pi(M)P(M,M')} \sum_{I:|I|\leq k} \sum_{F:|F|\geq k-1}\pi(I)\pi(F) |\gamma_{IF}^{(i)}|^2 \\
    \leq &2\cdot mn^2(1+\lambda^2) \cdot (2n)^2\cdot N_{k,k-1} \cdot \frac{1}{N_{k,k-1}} \cdot\frac{\sum_{l=1}^{k} N_l\cdot \lambda^{2l}l! \cdot \sum_{l=k-1}^n N_l\cdot \lambda^{2l}l!}{\lambda^{2k}k!\cdot \sum_{l=1}^n N_l\cdot \lambda^{2l}l!}.
\end{align}
Similarly, when $\lambda>\frac{1}{2}$,
\begin{align}
    &\frac{\sum_{l=1}^{k} N_l\cdot \lambda^{2l}l!}{N_{k,k-1} \lambda^{2k}k!}=\frac{1}{k} \cdot \frac{\sum_{l=1}^{k} N_l\cdot \lambda^{2l}l!}{N_{k}\cdot \lambda^{2k}k!}
    \leq \frac{1}{k}\cdot \left(\frac{4\lambda^2}{4\lambda^2-1}+1\right)\leq \frac{4\lambda^2}{4\lambda^2-1}+1.
\end{align}
Considering these two cases
we have 
\begin{align}
    \varrho_t&\leq \max\left\{2\cdot \frac{mn(1+\lambda^2)}{\lambda^2} \cdot (2n)^2\cdot\lambda^2\left(\frac{4\lambda^2}{4\lambda^2-1}+1\right),2\cdot mn^2(1+\lambda^2) \cdot (2n)^2\cdot\left(\frac{4\lambda^2}{4\lambda^2-1}+1\right)\right\}\\
    &= 8 mn^4\cdot \frac{(\lambda^2+1)(8\lambda^2-1)}{4\lambda^2 - 1}.
\end{align}

Let $\bar{\lambda}=\max\{1/\lambda,\lambda\}$, noting that $\lambda^{2|M|}\cdot |M|!\leq \bar{\lambda}^{2n}\cdot n!$ for any matching $M$, and the number of matchings in the complete bipartite graph $G$ is bounded by $\frac{m!}{(m-n)!}\cdot n!\leq (m!)^2 \leq m^{2m}$. We have
the crude bound $\ln\pi(M)^{-1} \leq \ln{\left({\bar{\lambda}}^{2n}\cdot n!\cdot m^{2m} \right)}\leq 2n \ln |{\lambda}| + 2m\ln{m}$, which holds uniformly over all matchings of $G$, by \thm{canonical-mixing-multiple}, we establish a polynomial mixing time of double-loop Glauber on complete bipartite graphs:
\begin{lemma}\label{lem:poly-mix}
     The mixing time of double-loop Glauber on the complete bipartite graph $G_{m,n}$ is less than $r_\epsilon = 16 mn^4\cdot \frac{(\lambda^2+1)(8\lambda^2-1)}{4\lambda^2 - 1}\cdot \left(\ln{\epsilon^{-1}}+ 2n \ln |{\lambda}| + 2m\ln{m}\right)$.
\end{lemma}

Finally, we note that the uniform sampling of all matchings guaranteed by \lem{bigraph-uniform} is subject to approximation errors. Consequently, when computing the total time needed by double-loop Glauber, we must incorporate these error terms into our analysis. Since the maximum matching size in $G_{m,n}$ is at most $n$, each inner layer sampling step is performed on a balanced bipartite graph with at most $2n$ vertices.

Take 
\begin{align}
\eta_{\frac{\eps}{2}} =\frac{\eps}{2}/r_{\frac{\eps}{2}} = \frac{\epsilon}{2}\left(16 mn^4\frac{(\lambda^2+1)(8\lambda^2-1)}{4\lambda^2 - 1}\left(\ln{{\epsilon}^{-1}}+ 2 n \ln |{\lambda}| + 2m\ln{m}\right)\right)^{-1}
\end{align}
in \lem{bigraph-uniform}. Then the failure probability of inner layer sampling in each iteration can be bounded within $\eta_{\frac{\eps}{2}}$.
Again by \lem{poly-mix}, after $r_{\frac{\eps}{2}}$ iterations, the total variance can be bounded by $\frac{\epsilon}{2} + \eta_{\frac{\eps}{2}} \cdot r_{\frac{\eps}{2}}=\epsilon$. Therefore, we give a polynomial-time upper bound for the mixing of double-loop Glauber on complete bipartite graph $G_{m,n}$:

\begin{theorem}\label{thm:poly-mix}
    Given a bipartite graph $G_{m,n}$, then for $\lambda>\frac{1}{2}$, given error $\eps$, we can achieve a sampling using double-loop Glauber in time $\tilde{O}(m^2 n^{15}(\log \epsilon^{-1})^2)$ such that the total variation distance between the sampling distribution and the ideal stationary distribution is at most $\epsilon$. The iterations of double-loop Glauber is $r_{\frac{\eps}{2}}= \tilde{O}(m^2n^4 \log \eps^{-1})$, and each inner layer sampling procedure needs $O(n^{11}(\log n)^2 (\log n + \log r_{\frac{\eps}{2}} + \log \eps^{-1}))$ time. 
\end{theorem}

\subsection{Dense Bipartite Graphs}
\label{app:dense-bipartite}
In this section, we extend the results from \append{complete-bipartite} to non-complete bipartite graphs, providing proofs of \thm{bipartite-graph} and \thm{main-results-bigraph}. A natural approach is to leverage the aforementioned congestion calculations.
We discuss dense bipartite graphs $G$: there exists a constant $\xi$, such that each vertex is non-adjacent to at most $\xi$ vertices in the opposite partition.

For any transition $t=(M,M')$, let $\tilde{M}$ denote the matching of smaller size in $\{M,M'\}$. Since we have
\begin{align}
    \pi(M)P(M,M') = \frac{1}{|E|}\cdot\frac{\lambda^2}{1+\lambda^2}\cdot \lambda^{2|\tilde{M}|} \Haf(\tilde{M}),
\end{align}
we can rewrite \eq{congestion-matching} as follows:
\begin{align}
    \varrho := & \max_{t=(M,M')}\bigg\{\frac{1}{\pi(M)P(M,M')} \sum_{(I,F,i)\in cp(t)} p^{(i)}(I,F)\cdot \pi(I)\pi(F) |\gamma_{IF}^{(i)}|\bigg\}\\
     =& \max_{t=(M,M')}\bigg\{ |E|\frac{1+\lambda^2}{{\lambda}^2}\cdot\frac{\sum_{(I,F,i)\in cp(t)} p^{(i)}(I,F)\cdot \lambda^{2|I|}\Haf(I)\lambda^{2|F|}\Haf(F) |\gamma_{IF}^{(i)}|}{\lambda^{2|\tilde{M}|}\Haf(\tilde{M})\cdot \sum_{X\text{ is a matching}} \lambda^{2|X|}\Haf(X)} \bigg\}.
\end{align}
For any bipartite graph $G$, let $G'$ be the corresponding complete bipartite graph with the same vertex set as $G$.
Let $\varrho_t'$ denote the congestion of transition $t=(M,M')$ in $G'$, $E'$ denote the edge set of $G'$, $P'(M,M')$ denote the corresponding transition probability of double-loop Glauber for $G'$, and $\Haf'(M)$ denote the corresponding Hafnian of the subgraph induced by matching $M$ in $G'$, which is equal to $|M|!$.
We have already established an estimate for $\varrho_t'$ in \append{complete-bipartite}, and now we can use this estimate to bound $\varrho_t$ in $G$ by estimate $\Haf(M)$ via $\Haf'(M)$ for any matching $M$.

Note that for any two matchings $I$ and $F$, the canonical path from $I$ to $F$ remains identical in both $G$ and $G'$. Furthermore, for any matching $M$ of size $n'\geq \xi$, let $V'$ denote the vertex set of $M$. It is straightforward that the minimum degree $d'$ of the subgraph induced by $V'$ satisfies $d'\geq n' - \xi$. As a result, by~\cite{hall1948distinct} we can bound the number of perfect matching in this subgraph:

\begin{lemma}[\cite{hall1948distinct}]
    \label{lem:hall}
    For any balanced bipartite graph $G=(V_1,V_2,E)$ with $|V_1| = |V_2| = n$, if the minimum degree of $G$ is at least $d$, and there exists a perfect matching, then the number of perfect matchings in $G$ is at least $d!$.
\end{lemma}
Thus the number of perfect matching in this subgraph is not less than $(n'-\xi)!$, which means
\begin{align}
    1\geq \frac{\Haf(M)}{\Haf'(M)}\geq \frac{(n'-\xi)!}{n'!}\geq \frac{1}{n'^\xi}\geq \frac{1}{n^\xi}.\label{eq:ratio-1}
\end{align}

Therefore, we have
\begin{align}
    \varrho  =& \max_{t=(M,M')}\bigg\{\frac{\sum_{(I,F,i)\in cp(t)} p^{(i)}(I,F)\cdot \lambda^{2|I|}\Haf(I)\lambda^{2|F|}\Haf(F) |\gamma_{IF}^{(i)}|}{\lambda^{2|\tilde{M}|}\Haf(\tilde{M})\cdot \sum_{X\text{ is a matching}}\lambda^{2|X|} \Haf(X)} \cdot |E|\frac{1+\lambda^2}{{\lambda}^2}\bigg\}\\
    \leq &\max_{t=(M,M')}\bigg\{\frac{ \sum_{(I,F,i)\in cp(t)} p^{(i)}(I,F)\cdot \lambda^{2|I|}\Haf'(I)\lambda^{2|F|}\Haf'(F) |\gamma_{IF}^{(i)}|}{\lambda^{2|\tilde{M}|}\Haf'(\tilde{M})n^{-\xi}\cdot \sum_{X\text{ is a matching}} \lambda^{2|X|}\Haf'(X)n^{-\xi}}\cdot |E'|\frac{1+\lambda^2}{{\lambda}^2}\bigg\}\\
    \leq & n^{2\xi} \varrho'.
\end{align}

\begin{proof}[Proof of \thm{bipartite-graph}]
    By the analysis of \lem{poly-mix}, we have the congestion $\varrho'$ of double-loop Glauber on complete bipartite graph $G_{m,n}$ is bounded by $\varrho' = 8 mn^4\cdot \frac{(\lambda^2+1)(8\lambda^2-1)}{4\lambda^2 - 1}$. 
    Then we have the congestion $\varrho$ of double-loop Glauber on bipartite graph $G$ is bounded by $\varrho \leq n^{2\xi} \cdot \varrho' = 8 mn^{4+2\xi}\cdot \frac{(\lambda^2+1)(8\lambda^2-1)}{4\lambda^2 - 1}$.
    Thus again by \thm{canonical-mixing-multiple}, the mixing time of double-loop Glauber on complete bipartite graph $G_{m,n}$ is less than $ 16 mn^{4+2\xi}\cdot \frac{(\lambda^2+1)(8\lambda^2-1)}{4\lambda^2 - 1}\cdot \left(\ln{\epsilon^{-1}}+ 2n \ln |{\lambda}| + 2m\ln{m}\right)$.
\end{proof}

Furthermore, we complete the proof of \thm{main-results-bigraph} by extending the analysis of \thm{poly-mix} to the case of dense bipartite graphs.
\begin{proof}[Proof of \thm{main-results-bigraph}]
Take 
\begin{align}
    \eta_{\frac{\eps}{2}} =\frac{\epsilon}{2}\left(16 mn^{4+2\xi}\frac{(\lambda^2+1)(8\lambda^2-1)}{4\lambda^2 - 1}\left(\ln{{\epsilon}^{-1}}+ 2 n \ln |{\lambda}| + 2m\ln{m}\right)\right)^{-1}
\end{align}
in \lem{bigraph-uniform}. Then the failure probability of each inner layer sampling step can be bounded within $\eta_{\frac{\eps}{2}}$.
Again by \lem{poly-mix}, after
\begin{align}
16 mn^{4+2\xi}\frac{(\lambda^2+1)(8\lambda^2-1)}{4\lambda^2 - 1}\left(\ln{{\epsilon}^{-1}}+ 2 n \ln |{\lambda}| + 2m\ln{m}\right)
\end{align}
iterations, the total variance can be bounded by $\epsilon$. Therefore, we give a polynomial-time upper bound $\tilde{O}{(m^2n^{(15+2\xi)}(\log\eps^{-1})^2)}$ for the mixing time of double-loop Glauber on the given dense bipartite graph.
\end{proof}


\subsection{Complete Graph $G_{2n}$}
\label{app:complete-non-bipartite}
In this section we turn our analysis to complete graph with $2n$ vertices, where the fundamental approach remains consistent with our complete bipartite graph methodology. 

We can again employ the same multiple canonical path scheme used for complete bipartite graphs. This approach still retains two key properties for complete graph:
\begin{itemize}
    \item All paths from $I$ to $F$ in this scheme only traverse matchings whose sizes lie within $\min\{|I|,|F|\} -1,\leq \max\{|I|,|F|\} + 1$ exactly as in the complete bipartite case; and the length of each path is bounded by $2n$.
    \item For complete graph $G_{2n}$, symmetry condition stated in \prop{congestion-symmetry-bipartite} still holds. We restate it here for clarity:
\end{itemize}
\begin{proposition}\label{prop:congestion-symmetry-non-bipartite}
    For complete graphs, under the aforementioned multiple canonical paths construction in \append{complete-bipartite}, the congestion of each transition in $T^+(k)$ (or $T^-(k)$) is identical.
\end{proposition}

Now we estimate the congestion of one transition $t$ in complete graph $G_{2n}$. First, we calculate $\pi(M)$ for each matching. Since any matching $M$ expands a complete subgraph of size $2|M|$, we have $\Haf(M) = (2|M|-1)!!$, and thus $\pi(M)\propto \lambda^{2|M|} (2|M|-1)!!$.
Then for edge sets of any two matchings $M$ and $M \cup\{e\}$, we have
\begin{align}
    \Pr[M\to M\cup \{e\}] &=  \frac{1}{n(2n-1)}\cdot \frac{\lambda^2}{1+\lambda^2},\\
    \Pr[M \cup \{e\}\to M] &= \frac{1}{n(2n-1)}\cdot \frac{\Haf(G_M)}{\Haf(G_{M\cup\{e\}})}\cdot \frac{1}{1+\lambda^2} \geq \frac{1}{n(2n-1)}\cdot \frac{1}{2|M|+1}\cdot \frac{1}{1+\lambda^2}.
\end{align}

Then we recalculate parameter $N_p$ and $N_{p,q}$ of complete graph $G_{2n}$:
\begin{align}
    N_p =& \binom{2n}{2p}\cdot (2p-1)!! =\binom{2n}{2p} \cdot \frac{(2p)!}{2^p p!}=\frac{(2n)!}{p!(2n-2p)!\cdot2^p}
    \\ N_{p,q} =& \begin{cases}
    N_p\cdot \binom{2n-2p}{2} &\text{ if } q = p + 1;\\
    N_p\cdot p &\text{ if } q = p - 1.
\end{cases}
\end{align}
Now consider $t\in T^+(k)$. Suppose $t=(M,M')$ with $|M| = k$ and $|M'|=k+1$, we have 
\begin{align}
    \quad N_{k,k+1}\cdot \varrho_t 
   \leq& 2\cdot  \frac{1}{\pi(M)P(M,M')} \sum_{I:|I|\leq k+1} \sum_{F:|F|\geq k-1}\pi(I)\pi(F) |\gamma_{IF}^{(i)}|^2 \\
    \leq& 2\cdot \frac{n(2n-1)(1+\lambda^2)}{\lambda^2} \cdot (2n)^2\cdot N_{k,k+1} \cdot \frac{1}{N_{k,k+1}} \cdot \frac{\sum_{l=1}^{k+1} N_l\cdot \lambda^{2l} l! \cdot \sum_{l=k-1}^n N_l\cdot \lambda^{2l} l!}{\lambda^{2k}k!\cdot \sum_{l=1}^n N_l\cdot \lambda^{2l} l!}.
\end{align}
Similarly,
notice that $\frac{N_{l+1}\cdot \lambda^{2(l+1)}(l+1)!}{N_l\cdot \lambda^{2l}l!} = \frac{(2n-2l)(2n-2l-1)}{2(l+1)}\cdot \lambda^2\cdot (l+1)\geq 6\lambda^2$ for $l<n-1$, when $\lambda^2>\frac{1}{6}$, we have
\begin{align}
    \frac{\sum_{l=1}^{k+1} N_l\cdot \lambda^{2l}l!}{N_{k,k+1} \lambda^{2k}k!} 
    =&\frac{2}{(2n-2k)(2n-2k-1)} \cdot \frac{\sum_{l=1}^{k+1} N_l \lambda^{2l}l!}{N_{k+1}\cdot \lambda^{2(k+1)}(k+1)!}\cdot \frac{N_{k+1}\lambda^2(k+1)}{N_k}\\
    \leq& \frac{2}{(2n-2k)(2n-2k-1)} \cdot \left(\frac{6\lambda^2}{6\lambda^2-1}+1\right)\cdot 
    \frac{1}{2}\lambda^2 (2n-2k)(2n-2k-1) \\=&\lambda^2\left(\frac{6\lambda^2}{6\lambda^2-1}+1\right).
\end{align}
Then consider $t\in T^-(k)$. Suppose $t=(M,M')$ with $|M| = k$ and $|M'|=k-1$, through similar analysis we have
\begin{align}
 N_{k,k-1}\cdot \varrho_t 
 \leq &2\cdot  \frac{1}{\pi(M)P(M,M')} \sum_{I:|I|\leq k} \sum_{F:|F|\geq k-1}\pi(I)\pi(F) |\gamma_{IF}^{(i)}|^2 \\
    \leq &2\cdot n(2n-1)(2n+1)(1+\lambda^2) \cdot (2n)^2\cdot N_{k,k-1} \cdot \frac{1}{N_{k,k-1}} \cdot\frac{\sum_{l=1}^{k} N_l\cdot \lambda^{2l}l! \cdot \sum_{l=k-1}^n N_l\cdot \lambda^{2l}l!}{\lambda^{2k}k!\cdot \sum_{l=1}^n N_l\cdot \lambda^{2l}l!}.
\end{align}
Similarly, when $\lambda^2>\frac{1}{6}$,
\begin{align}
    \frac{\sum_{l=1}^{k} N_l\cdot \lambda^{2l}l!}{N_{k,k-1} \lambda^{2k}k!}
    =&\frac{1}{k} \cdot \frac{\sum_{l=1}^{k} N_l\cdot \lambda^{2l}l!}{N_{k}\cdot \lambda^{2k}k!}
    \leq \frac{1}{k}\cdot \left(\frac{6\lambda62}{6\lambda^2-1}+1\right)\leq \frac{6\lambda^2}{6\lambda^2-1}+1.
\end{align}

Considering these two cases we have 
\begin{align}
    \varrho_t&\leq \max\left\{2 \frac{n(2n-1)(1+\lambda^2)}{\lambda^2}  (2n)^2\lambda^2\left(\frac{6\lambda^2}{6\lambda^2-1}+1\right),2n(2n-1)(2n+1)(1+\lambda^2)  (2n)^2\left(\frac{6\lambda^2}{6\lambda^2-1}+1\right)\right\}\\
    &\leq 32 n^5 \frac{(\lambda^2+1)(12\lambda^2-1)}{6\lambda^2 - 1}.
\end{align}

Let $\bar{\lambda}=\max\{1/\lambda,\lambda\}$, then we have the crude bound $\ln\pi(M)^{-1} \leq \ln{\left({\bar{\lambda}}^n\cdot (2n-1)!!\cdot(2n+1)^{2n} \right)}\leq n \ln |{\lambda}| + 3n\ln{(2n+1)}$, which holds uniformly over all matchings of $G$.

By \thm{canonical-mixing-multiple}, we establish a polynomial mixing time of double-loop Glauber on complete graphs:
\begin{lemma}\label{lem:poly-mix-non-bipartite}
     The mixing time of double-loop Glauber on complete graph $G_{2n}$ is less than $r_\epsilon = 64 n^5\cdot \frac{(\lambda^2+1)(12\lambda^2-1)}{6\lambda^2 - 1}\cdot \left(\ln{\epsilon^{-1}}+ n \ln |{\lambda}| + 3n\ln{(2n+1)}\right)$.
\end{lemma}

Finally, we incorporate the error terms from the uniform sampling of all matchings guaranteed by \lem{non-bigraph-uniform} into our analysis. Since the maximum matching size in $G_{2n}$ is at most $n$, each inner layer sampling step is performed on a complete graph with at most $2n$ vertices.

Take 
\begin{align}
\eta_{\frac{\eps}{2}} =\frac{\eps}{2}/r_{\frac{\eps}{2}} = \frac{\epsilon}{2}\left(64 n^5\cdot \frac{(\lambda^2+1)(12\lambda^2-1)}{6\lambda^2 - 1}\cdot \left(\ln{\epsilon^{-1}}+ n \ln |{\lambda}| + 3n\ln{(2n+1)}\right)\right)^{-1}
\end{align}
in \lem{non-bigraph-uniform}. Then the failure probability of the inner layer sampling in each iteration can be bounded within $\eta_{\frac{\eps}{2}}$.
Again by \lem{poly-mix-non-bipartite}, after $r_{\frac{\eps}{2}}$ iterations, the total variance can be bounded by $\frac{\epsilon}{2} + \eta_{\frac{\eps}{2}} \cdot r_{\frac{\eps}{2}}=\epsilon$. Therefore, we give a polynomial-time upper bound for the mixing time of the double-loop Glauber dynamics on the complete graph $G_{2n}$:

\begin{theorem}\label{thm:poly-mix-non-bipartite}
    Given a graph $G_{2n}$, then for $\lambda>\frac{1}{\sqrt{6}}$, given error $\eps$, we can achieve a sampling using double-loop Glauber in time $\tilde{O}( n^{20}(\log \epsilon^{-1})^3)$ such that the total variation distance between the sampling distribution and the ideal stationary distribution is at most $\epsilon$. The iterations of double-loop Glauber is $r_{\frac{\eps}{2}}= \tilde{O}(n^6 \log \eps^{-1})$, and each inner layer sampling procedure needs $\tilde{O}(n^{14}(\log{\eps^{-1}})^2)$ time.
\end{theorem}

\subsection{Dense Non-bipartite Graphs}
\label{app:dense-non-bipartite}
In this subsection, we explore extending the results from \append{complete-non-bipartite} to non-complete graphs, providing proofs of \thm{non-bipartite-graph} and \thm{main-results-non-bigraph}.
Similar as \append{dense-bipartite}, we discuss some dense graphs $G$: there exists a constant $\xi$, such that each vertex is non-adjacent to at most $\xi$ vertices.

For any graph $G$, let $G'$ be the corresponding complete bipartite graph with the same vertex set as $G$.
Let $\varrho_t'$ denote the congestion of transition $t=(M,M')$ in $G'$, $P(M,M')$ denote the corresponding transition probability of double-loop Glauber for $G'$, and $\Haf'(M)$ denote the corresponding Hafnian of the subgraph induced by matching $M$ in $G'$, which is equal to $(2|M|-1)!!$.
Then we use the estimate of $\varrho_t'$ given in \append{complete-non-bipartite} to bound $\varrho_t$ in the original graph $G$.

Note that for any two matchings $I$ and $F$, the canonical path from $I$ to $F$ remains identical in both $G$ and $G'$. 

In what follows, we employ graph-theoretic techniques to bound the number of perfect matchings in subgraphs induced by matchings. 
\begin{proposition}
    Let $G_E$ denote the induced graph of a matching $M$ of size $n'\geq \xi$, then the number of perfect matchings of is at least 
    \begin{align}
        \frac{(2n'-1-\xi)!!}{(\xi+1)!!}.
     \end{align} 
\end{proposition}
\begin{proof}
    Our analysis relies primarily on the Ore's Theorem~\cite{ore1960note}: 
\begin{lemma}
    Let $G=(V,E)$ be a simple graph. If for every pair of distinct non-adjacent vertices $u$ and $v$, the sum of their degrees satisfies $\deg(u) + \deg(v) \geq |V|$, then $G$ contains a Hamilton cycle.
\end{lemma}
    Ore's theorem shows that if the minimum degree of a graph with even vertices is at least $|V|/2$, then the graph contains a Hamilton cycle, which is also a perfect matching.

    Note that the minimum degree $d'$ of the subgraph $G_E$ satisfies $d'\geq 2n' - 1 - \xi$. We start from a vertex $u_1$, which has $2n' - 1 - \xi$ possible pairing choices with other vertices. For the choice of pair $(u_1,v_1)$, we delete the edge $(u_1,v_1)$ and all edges incident to $u_1$ and $v_1$. The remaining graph has $2n' - 2 - \xi$ vertices, and the minimum degree of this graph is at least $2n' - 3 - \xi$. We can repeat this process and find another pair $(u_2,v_2)$, and delete the edge $(u_2,v_2)$ and all edges incident to $u_2$ and $v_2$. This procedure can be repeated until we find $n'-\xi - 1$ pairs of vertices, and then the remaining graph has $2\xi+1$ vertices, with the minimum degree of the current graph being at least $2n' - 1 - \xi - 2(n'-\xi - 1) = \xi + 1$. Thus again by Ore's theorem, the remaining graph contains a Hamilton cycle. Finally, we choose this Hamilton cycle and find a perfect matching of $G_E$. The number of choices is
    \begin{align}
        (2n'-\xi-1)\cdot(2n'-\xi-3)\cdots (\xi+3) = \frac{(2n'-1-\xi)!!}{(\xi+1)!!}.
    \end{align}
\end{proof}

Now, we can bound $\Haf(M)$ by $\Haf'(M)$ for any matching $M$:
\begin{align}\label{eq:ratio-2}
1\geq \frac{\Haf(M)}{\Haf'(M)}\geq \frac{{(2n'-1-\xi)!!}/{(\xi+1)!!}}{(2n'-1)!!}\geq \left((2n')^{\frac{\xi+1}{2}}\cdot (\xi+1)!!\right)^{-1}
\geq \frac{1}{(2n)^\xi}.
\end{align}

Therefore, we have
\begin{align}
    \varrho  =& \max_{t=(M,M')}\bigg\{\frac{\sum_{(I,F,i)\in cp(t)} p^{(i)}(I,F)\cdot \lambda^{2|I|}\Haf(I)\lambda^{2|F|}\Haf(F) |\gamma_{IF}^{(i)}|}{\lambda^{2|\tilde{M}|}\Haf(\tilde{M})\cdot \sum_{X\text{ is a matching}}\lambda^{2|X|} \Haf(X)} \cdot |E|\frac{1+\lambda^2}{{\lambda}^2}\bigg\}\\
    \leq &\max_{t=(M,M')}\bigg\{\frac{\sum_{(I,F,i)\in cp(t)} p^{(i)}(I,F)\cdot \lambda^{2|I|}\Haf'(I)\lambda^{2|F|}\Haf'(F) |\gamma_{IF}^{(i)}|}{\lambda^{2|\tilde{M}|}\Haf'(\tilde{M})(2n)^{-\xi}\cdot \sum_{X\text{ is a matching}}\lambda^{2|X|} \Haf'(X)(2n)^{-\xi}} \cdot |E'|\frac{1+\lambda^2}{{\lambda}^2}\bigg\}\\
    \leq & {(2n)}^{2\xi} \varrho'.
\end{align}

\begin{proof}[Proof of \thm{non-bipartite-graph}]
    By the analysis of \lem{poly-mix-non-bipartite}, we have the congestion $\varrho'$ of double-loop Glauber on complete graph $G_{2n}$ is bounded by $\varrho' = 32 n^5\cdot \frac{(\lambda^2+1)(12\lambda^2-1)}{6\lambda^2 - 1}$. 
    Then we have the congestion $\varrho$ of double-loop Glauber on bipartite graph $G$ is bounded by $\varrho \leq {2n}^{2\xi} \cdot \varrho' = 32\cdot 2^\xi \cdot n^{5+2\xi}\cdot \frac{(\lambda^2+1)(12\lambda^2-1)}{6\lambda^2 - 1}$.
    Thus by \thm{canonical-mixing-multiple}, the mixing time of double-loop Glauber on complete graph $G_{2n}$ is less than $ 64 \cdot 2^\xi\cdot n^{5+2\xi}\cdot \frac{(\lambda^2+1)(12\lambda^2-1)}{6\lambda^2 - 1}\cdot \left(\ln{\epsilon^{-1}}+ n \ln |{\lambda}| + 3n\ln{(2n+1)}\right)$.
\end{proof}

Finally, we complete the proof of \thm{main-results-non-bigraph} by extending the analysis of \thm{poly-mix-non-bipartite} to the case of dense graphs.
\begin{proof}[Proof of \thm{main-results-non-bigraph}]
Take
\begin{align}
\eta_{\frac{\eps}{2}} = \frac{\epsilon}{2}\left(64 \cdot 2^\xi\cdot n^{5+2\xi}\cdot \frac{(\lambda^2+1)(12\lambda^2-1)}{6\lambda^2 - 1}\cdot \left(\ln{\epsilon^{-1}}+ n \ln |{\lambda}| + 3n\ln{(2n+1)}\right)\right)^{-1}
\end{align}
in \lem{bigraph-uniform}. Then the failure probability of each inner layer sampling step can be bounded within $\eta_{\frac{\eps}{2}}$, and need $\tilde{O}(n^{14}(\log{\eps^{-1}})^2)$ time.
Again by \lem{poly-mix-non-bipartite}, after
\begin{align}
64 \cdot 2^\xi\cdot n^{5+2\xi}\cdot \frac{(\lambda^2+1)(12\lambda^2-1)}{6\lambda^2 - 1}\cdot \left(\ln{\epsilon^{-1}}+ n \ln |{\lambda}| + 3n\ln{(2n+1)}\right)
\end{align}
iterations, the total variance can be bounded by $\epsilon$. Therefore, we give a polynomial-time upper bound $\tilde{O}{(n^{(20+2\xi)}(\log\eps^{-1})^3)}$ for the mixing time of the double-loop Glauber dynamics on the given dense graph.
\end{proof}

\subsection{Technical Barrier for Non-dense Graphs} \label{app:non-dense-graph}

In this subsection, we highlight several technical challenges that arise when extending our approach to non-dense graphs.

For the proof for dense graphs in \append{dense-bipartite} and \append{dense-non-bipartite}, the key insight of our proof relies on bounding the Hafnian of matchings in the given dense graph via the Hafnian of matchings in corresponding complete graphs, as demonstrated in \eq{ratio-1} and \eq{ratio-2}. For non-dense graphs, unfortunately the ratio between the Hafnian in the current graph and that in complete graphs can become exponentially large, particularly for sparse graphs. Consequently, directly applying results from complete graphs fails to yield polynomial mixing time.

This observation exposes a fundamental limitation of the canonical path method in our setting. The Hafnian intrinsically corresponds to the weight of matchings in the stationary distribution. While for dense graphs we effectively bound these weights using  Hafnian in corresponding complete graphs, for non-dense graphs, even matchings with a large number of edges can have very small Hafnian. Consequently, such low-weight matchings are likely to become the bottleneck in the canonical path construction, inducing large congestion.

Here we consider a potential hard instance shown in \sfig{hard-instance}\blue{a}. The graph consists of several squares arranged in a cycle. Label the four vertices of each square clockwise, and connect the first vertex of each square to its third vertex with an edge, and connect the second vertex of each square to the fourth vertex of the previous square with an edge.

\begin{figure*}[!htbp]
    \centering
    \includegraphics[width=.9\linewidth]{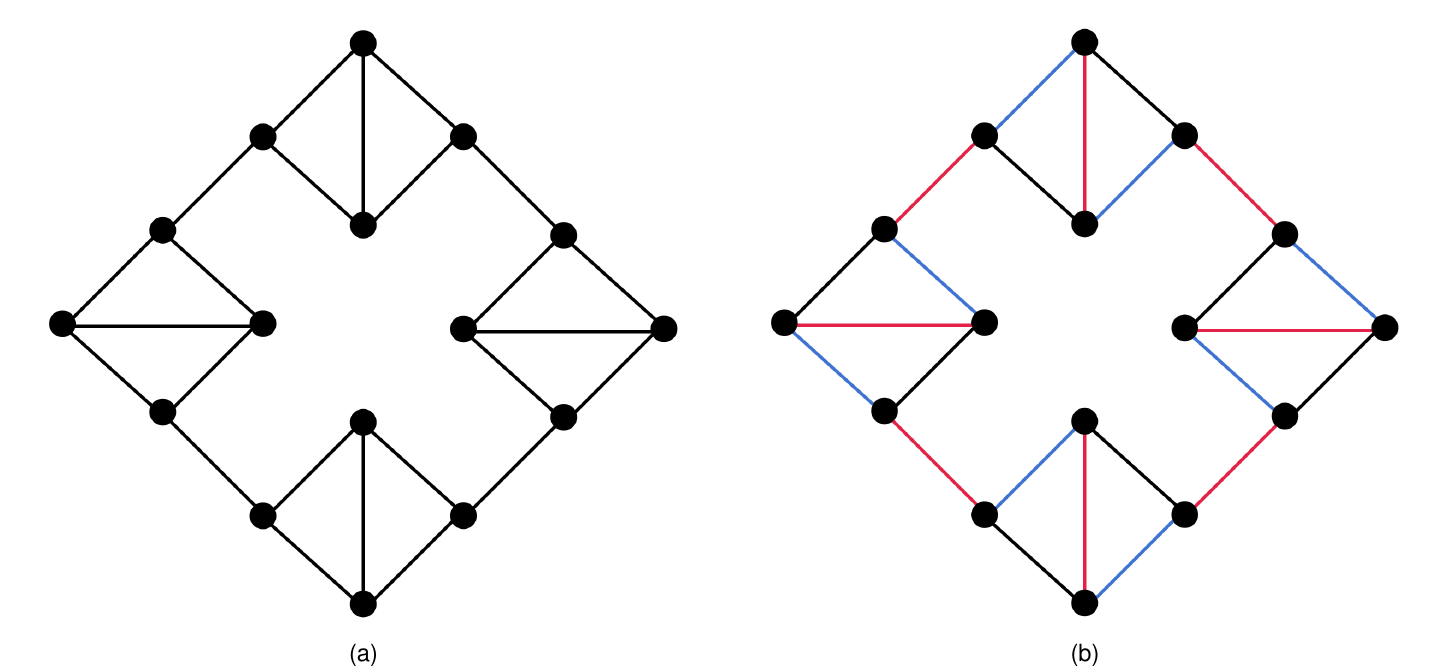}
    \caption{A hard instance for non-dense graphs and its two different types of perfect matchings. The graph in the left figure consists of several squares arranged in a cycle, where the first vertex of each square is connected to its third vertex with an edge, and the second vertex of each square is connected to the fourth vertex of the previous square by an edge. The red edges in the right figure form a perfect matching $M_0$ in set $A$ (which is also the unique one); the blue edges form a perfect matching $M_1$ in set $B$. Notice that deleting any edge in $M_0$ will result in a near perfect matching with Hafnian $1$.}
    \label{sfig:hard-instance}
\end{figure*}

Notice that the perfect matchings of this graph can be divided into two categories: those that include the edge connecting different squares (denoted as set $A$), and those that only include edges within each square (denoted as set $B$). The first category of perfect matchings has only one perfect matching $M_0$, as is shown in \sfig{hard-instance}\blue{b}. The second category of perfect matchings has $2^{n}$ perfect matchings, where $n$ is the number of squares in the graph. This is because each square has two choices for selecting edges, and there are $n$ squares in total. Therefore the Hafnian of any perfect matching is $1+2^n$. Now consider the canonical paths from $B$ to $A$. The final step of these canonical paths must be adding one edge from a near perfect matching and get $M_0$, which have $2n$ choices. A significant insight is that these $2n$ near perfect matchings are all very low-weight matchings, with Hafnian being $1$. Then we can consider the sum of congestion of these transitions:

\begin{align}
    \sum_{e\in M_0} \rho_{t=(M_0\backslash \{e\}, M_0)} \geq &\sum_{e\in M_0} \frac{1}{\pi(M_0\backslash\{e\})P(M_0\backslash \{e\},M_0)} \cdot \sum_{\substack{(I,F,i)\in cp(t) \\ 
    I\in B, F\in A}} p^{(i)}(I,F)\cdot \pi(I)\pi(F) |\gamma_{IF}^{(i)}|\\
    \geq & \frac{1}{\pi(M_0\backslash{\{e_0\}})\cdot P(M_0\backslash \{e_0\},M_0)} \cdot \sum_{I\in B}\pi(I)\pi(M_0)\cdot n \\
    \geq & \frac{\pi(M_0)}{\pi(M_0\backslash{\{e_0\}})}\cdot \frac{1}{\frac{1}{6n}\cdot \frac{\lambda^2}{1+\lambda^2}}\cdot n\cdot \sum_{I\in B} \pi(I),
\end{align}

where $e_0$ is one edge in $M_0$. By taking $\lambda = n$ we ensure that the stationary distribution's weight is concentrated on perfect matchings, which means that the last term in the above expression is at the polynomial level. Since $\frac{\pi(M_0)}{\pi(M_0\backslash{\{e_0\}})} = {2^n}$, the maximal congestion must be exponentially large, and the canonical path method fails.

Furthermore, we can directly give a lower bound of the mixing time of this hard instance by conductance.

\begin{lemma}[Theorem 7.4,~\cite{levin2017markov}] For a given ergodic, aperiodic, and irreducible Markov chain, let $\Phi$ denote the conductance defined in \eq{conductance}, we have

\begin{align}
    t_{\mix} = t_{\mix} \left(\frac{1}{4}\right) \geq \frac{1}{4\Phi}.
\end{align}
\end{lemma}

The conductance of the set $S$ is defined to be
\begin{align}
    \Phi(S) :=  \frac{\sum_{i\in S, j\not\in S}p_{ij}\pi_i}{\sum_{i\in S}\pi_i},
\end{align}
and we have
\begin{align}
    \Phi = \min_{S:0< \sum_{i\in S} \pi_i\leq \frac{1}{2}} \Phi(S).
\end{align}
Therefore, to establish a lower bound on the mixing time, we only need to find a set $S$ such that its conductance is quite small. Just take $S=A$.
\begin{align}
    \Phi(A) =& \frac{1}{\pi(M_0)}\cdot \sum_{e'\in M_0} \pi(M_0)P(M_0,M_0\backslash\{e'\})\\
    =& \sum_{e'\in M_0} P(M_0,M_0\backslash\{e'\})\\
    =&\sum_{e'\in M_0} \Pr[\text{the chosen edge is $e'$ in \lin{double-loop-4} in \algo{double-loop}}]\cdot \Pr[\text{$e\in E_t$ in \lin{double-loop-8} in \algo{double-loop}}]\cdot \frac{1}{1+\lambda^2}\\
    =& 2n \cdot \frac{1}{6n} \cdot \frac{1}{1+2^n}\cdot \frac{1}{1+\lambda^2}.
\end{align}

This means that for the given hard instance, the double-loop Markov chain satisfies
\begin{align}
    t_{\mix} \geq \frac{1}{4\Phi} \geq \frac{1}{4\Phi(A)} = 3(1+\lambda^2)\cdot (1+2^n),
\end{align}
which is exponentially large.

For non-dense non-bipartite graphs, there is an another additional technical challenge: the task of uniformly sampling perfect matchings from subgraphs in the inner loop is quite challenging for general non-bipartite graphs.
  
\section{Double-loop Glauber Dynamics on Weighted Graphs}\label{app:double-loop-weighted}
The original double-loop Glauber dynamics is designed for unweighted graphs, and this section extends the double-loop Glauber dynamics to weighted graphs.

First, we provide the formal definition of Hafnian on weighted graphs. 
\begin{definition}
    For a graph $G=(V,E)$, with $V=\{u_1, u_2,\ldots, u_{2n}\}$ and a weight $w_e$ assigned to each edge $e$, let matrix $A$ denote the weighted adjacency matrix of $G$ with
\begin{align}
    A_{i,j} = \begin{cases} 
    w_{(i,j)} & \text{ if } (u_i,u_j)\in E;\\
    0 & \text{ otherwise.}
    \end{cases}
\end{align}
Then the Hafnian of $G$ is defined as the Hafnian of $A$, as shown in \eq{Haf}:
    \begin{align}
    \Haf(G) :=\operatorname{Haf}(A)=\frac{1}{2^n n!} \sum_{\sigma \in \mathcal{S}_{2 n}} \prod_{i=1}^n A_{\sigma(2 i-1), \sigma(2 i)}.
    \end{align}
The above definition of Hafnian on weighted graphs can be rewritten as
\begin{align}
    \Haf(G) = \sum_{M\text{ is a perfect matching of } G} \prod_{e\in M} w_e.
\end{align}
\end{definition}

For convenience, let $w(M) = \prod_{e\in M} w_e$ denote the weight of matching $M$.
Same as unweighted graphs, the Hafnian of a vertex set in $G$ is defined as the Hafnian of the subgraph induced by the vertex set. Our goal is to sample a vertex set $S$ from $G$ such that
\begin{align}
    \Pr[S] \propto \lambda^{|S|}\Haf^2(S).
\end{align}

Without loss of generality, we can assume that $w_e\geq 1$ for $\forall e \in E$. Otherwise, suppose $w_* = \min_e\{w_e\}\leq 1$, then we can define a new weight $w'_e = w_e/w_*$ for all edges $e\in E$, and then take $\lambda_*=\lambda\cdot{w_*}$. The Hafnian of a vertex set with new weights is $\Haf_*(S) = w_*^{-\frac{|S|}{2}} \cdot \Haf(G')$, which means $\lambda^{|S|}\Haf^2(S) = \lambda_*^{|S|}\Haf_*^2(S)$.

Our double-loop Glauber dynamics on weighted graphs is similar to the unweighted case, but with a few modifications. In particular, the weight of a matching is defined as the product of the weights of its edges, and the probability of sampling $M$ is proportional to the weight of $M$, denoted as $w(M)$. The algorithm is formally presented in \algo{double-loop-weighted}. 

\begin{algorithm}[h]
    \SetAlgoLined 
	\KwIn{A graph $G=(V,E)$, with a weight $w_e\geq 1$ assigned to each edge $e$.}
	\KwOut{A sample of vertex set $S$ of $G$ such that
        $\Pr[S] \propto \lambda^{|S|}\Haf^2(S)$. }

        \vspace{1em}
        Initialize $X_0$ as an arbitrary matching in $G$;\
        
        Initialize $t\leftarrow 0$, set the number of iterations $T$;\
	
	\lWhile{$t < T$}\
    {
        \Indp Choose a uniformly random edge $e$ from $E$;\

        \lIf{$e$ and $X_t$ form a new matching}\
        {\Indp Set $X_{t+1} = X_t\cup\{e\}$ with probability $\frac{\lambda^2}{1+\lambda^2}$ and otherwise set $X_{t+1} = X_t$;\;}
        
        \lElse 
        {\lIf{$e$ is in $X_t$}\
         {\Indp Sample a perfect matching $E_i$ in the subgraph induced by $X_t$ with probability $\Pr[E_i]\propto w(E_i)$. If $e\not\in E_i$, set $X_{t+1} = X_t$. If $e\in E_i$, set $X_{t+1} = X_t \backslash \{e\}$ with probability $\frac{1}{1+\lambda^2}\cdot \frac{1}{w_e^2}$ and otherwise set $X_{t+1} = X_t$;\;\label{lin:weighted-8}
         \Indm\lElse\
          {\quad \quad Set $X_{t+1} =X_t$;}}}
        $t \leftarrow t+1$;\;
        \caption{Double-loop Glauber Dynamics for Weighted Graphs}
    \label{algo:double-loop-weighted}
 }
 Output the vertex set in matching $X_T$;\;
\end{algorithm}

The main difference between \algo{double-loop-weighted} and \algo{double-loop} is that we need to sample a perfect matching in the subgraph induced by $X_t$ with probability $\Pr[E_i]\propto w(E_i)$, and the transition probability of sampling $X_{t+1}$ from $X_t$ is proportional to $\frac{1}{w_e^2}$. 

\begin{lemma}
    The stationary distribution of the double-loop Glauber dynamics on weighted graphs is given by
    \begin{align}
        \Pr[S] \propto \lambda^{|S|}\Haf^2(S).
    \end{align}
\end{lemma}
\begin{proof}
    For two matchings $X$ and $X\cup\{e\}$, we can calculate the transition probability as follows:
\begin{align}
    \Pr[X\to X\cup\{e\}] = \frac{\lambda^2}{1+\lambda^2} \cdot \frac{1}{|E|}.
\end{align}
If $X_t = X\cup\{e\}$, the probability $\Pr[e\in E_i]$ in \lin{weighted-8} in \algo{double-loop-weighted} can be rewritten as
\begin{align}
    \Pr[e\in E_i] = \frac{\sum_{E'\text{ is a perfect matching of }G_{X} }w(E')\cdot w_e}{\sum_{E'\text{ is a perfect matching of }G_{X\cup\{e\}}} w(E')} = w_e\cdot \frac{\Haf(X)}{\Haf(X\cup\{e\})}.
\end{align}
Thus we have
\begin{align}
    \Pr[X\cup\{e\} \to X] = \frac{1}{|E|} \cdot \frac{1}{w^2_e}\frac{1}{1+\lambda^2}\cdot w_e\frac{\Haf(X)}{\Haf(X\cup\{e\})} ,
\end{align}
where the third term comes from the probability of edge $e$ being in $E_i$. As a result, we can directly verify that the stationary distribution of matchings on weighted graphs is given by
\begin{align}
    \pi(X)\propto \lambda^{2|E|}\cdot \Haf(G_X) \cdot w(X),
\end{align}
and the distribution on vertex sets is 
\begin{align}
    \pi(S)\propto \lambda^{|S|}\cdot \Haf(S)\cdot \sum_{\text{$X$ is a perfect matching of $G_S$}} w(X) = \lambda^{|S|}\cdot \Haf^2(S).
\end{align}
\end{proof}

Finally, we give a Markov chain to satisfy the desired sampling of perfect matchings in \lin{weighted-8} in \algo{double-loop-weighted}. The method is based on the Markov chain for uniform sampling of perfect matchings by~\cite{jerrum1989approximating}, which is discussed in \append{uniform-sampling-pm}. The modified Markov chain is given as follows:

Let $\mathcal{N}$ denote the set of all perfect matchings and near-perfect matchings of $G$.
In any state $M\in\mathcal{N}$, choose an edge $e=(u,v)\in E$ uniformly at random and then
\begin{itemize}
    \item If $M\in M_n(G)$ and $e\in M$, then move to state $M' = M \backslash \{e\}$ with probability $\frac{1}{w_e}$;
    \item If $M\in M_{n-1}(G)$ and $u,v$ are unmatched in $M$, then move to state $M' = M\cup \{e\}$ with probability $1$;
    \item If $M\in M_{n-1}(G)$, $u$ is matched to $z$ in $M$ and $v$ is unmatched in $M$, then move to state $M' = M \cup \{e\} \backslash \{(z,u)\}$ with probability $\min\{1,\frac{w_e}{w_{(z,u)}}\}$; symmetrically, if $v$ is matched to $z$ and $u$ is unmatched, then move to state $M' = M \cup \{e\} \backslash \{(z,v)\}$ with probability $\min\{1,\frac{w_e}{w_{(z,v)}}\}$;
    \item In all other cases, do nothing.
\end{itemize}
\begin{lemma}
    The stationary distribution of the above Markov chain is given by
    \begin{align}
        \Pr[M]\propto w(M).
    \end{align}
\end{lemma}
\begin{proof}
    For two near-perfect matchings $M$ and $M'$ such that $M'= M\cup \{(u,v)\} \backslash \{(z,u)\}$, the transition probability is given by
\begin{align}
    \Pr[M\to M'] = \frac{1}{|E|}\cdot \min\left\{1,\frac{w_{(u,v)}}{w_{(z,u)}}\right\};\quad \Pr[M'\to M] = \frac{1}{|E|} \cdot \min\left\{1,\frac{w_{(z,u)}}{w_{(u,v)}}\right\}.
\end{align}
For a perfect matching $M$ and a near-perfect matching $M'$ such that $M'= M\backslash \{(u,v)\}$, the transition probability is given by
\begin{align}
    \Pr[M\to M'] = \frac{1}{|E|};\quad \Pr[M'\to M] = \frac{1}{|E|}\cdot \frac{1}{w_{(u,v)}}.
\end{align}
Therefore, the stationary distribution of the Markov chain is given by $\pi(M) \propto \prod_{e\in M}{w_e} = w(M).$
\end{proof}


\section{Algorithms with Pseudocodes}\label{app:algorithm}
This section introduces the classical algorithms and their variants enhanced by Glauber dynamics we used for numerical experiments in \sec{experiments}. In the actual code execution, we use the method in \append{uniform-sampling-dense} to uniformly sample perfect matching. During the iterations of the Glauber dynamics, we create a post-selection list. As long as the size of current edge set equals $k/2$, we append the edge set to the list. Each time we generate a vertex set of size $k$ according to the Glauber dynamics, we choose the last element in the list as the selected matching. 

Similarly, each time we generate a vertex set of size $k$ according to the quantum inspired classical algorithm by~\citet{Oh_2024}, we sample $k/2$ columns of $V$ with replacement and repeat until the row indexes of the nonzero elements form a collision-free set of size $k$.

We first present the pseudocode for the original random search (RS) in \algo{random-search}. In each iteration, we randomly pick $k$ vertices from all $n$ vertices of the graph. If the objective function (Hafnian or density) corresponding to the subgraph is larger than the former recorded best value, we substitute the recorded best with the new submatrix.

\vspace{3mm}
\begin{algorithm}[H]
	\KwIn{
    
    $\boldsymbol{A}$ : Input matrix of dimension $n\times n$;
    
$k$ : Dimension of submatrix;

$f$ : Objective function;

$N:$ Number of iterations;}
	\KwOut{The submatrix $\boldsymbol{A}_B$ with the largest $f(\boldsymbol{A}_B)$; the function value $f(\boldsymbol{A}_B)$}
\vspace{1em}
        Initialize $i=0$, $Best=0$;
        
        Initialize $\boldsymbol{A}_B=None$;

	\lWhile{$i < N$}\
    {
        \Indp Choose $\boldsymbol{A}_i$, a $k$-dimensional submatrix of $\boldsymbol{A}$ uniformly at random; 

        Calculate $f(\boldsymbol{A}_i)$;
        
        \lIf{$f(\boldsymbol{A}_i)>Best$}\
        {\quad \quad $Best=f(\boldsymbol{A}_i)$;}
        
        {\Indm \quad \quad \quad ~$\boldsymbol{A}_B=\boldsymbol{A}_i$;}
 }
 
 Output $\boldsymbol{A}_B$ and $f(\boldsymbol{A}_B)$
\caption{Random Search}
 \label{algo:random-search}
\end{algorithm}

Next, we substitute the uniform random update in each iteration with \algo{glauber}, \algo{double-loop}, \algo{jerrum-glauber}, and the quantum-inspired classical algorithm by~\citet{Oh_2024} to enhance the above original random search. The pseudocode is presented in \algo{random-search-enhanced}. 

\begin{algorithm}[h]
	\KwIn{
    
    $\boldsymbol{A}$ : Input matrix of dimension $n\times n$;
    
$k$ : Dimension of submatrix ;

$f$ : Objective function ;

$N:$ Number of iterations ;}
	\KwOut{The submatrix $\boldsymbol{A}_B$ with the largest $f(\boldsymbol{A}_B)$; the function value $f(\boldsymbol{A}_B)$}
\vspace{1em}
        Initialize $i=0$, $Best=0$;
        
        Initialize $\boldsymbol{A}_B=None$;

	\lWhile{$i < N$}\
    {
        \Indp Choose $\boldsymbol{A}_i$, a $k$-dimensional submatrix of $\boldsymbol{A}$ according to \algo{glauber} or \algo{double-loop} or \algo{jerrum-glauber}; 

        Calculate $f(\boldsymbol{A}_i)$;
        
        \lIf{$f(\boldsymbol{A}_i)>Best$}\
        {\quad \quad $Best=f(\boldsymbol{A}_i)$;}
        
        {\Indm \quad \quad \quad ~$\boldsymbol{A}_B=\boldsymbol{A}_i$;}
 }
 
 Output $\boldsymbol{A}_B$ and $f(\boldsymbol{A}_B)$
	\caption{Random Search Enhanced by Glauber Dynamics}
    \label{algo:random-search-enhanced}
\end{algorithm}

Instead of applying i.i.d.~samples in random search, we adopt a neighbor update strategy in simulated annealing (SA). In each iteration, we randomly generate an integer $m\in[0,k-1]$ and randomly choose $m$ of the former vertices to keep unchanged. Then we randomly sample $k-m$ vertices from the remaining $n-m$ vertices to form a new vertex set of size $k$, and update the set with probability assigned by a Metropolis filter. With the annealing parameter $0<\gamma<1$, the temperature approaches exponentially small as iteration increases and the update probability when $f(\boldsymbol{A}_R)<(\boldsymbol{A}_S)$ approaches $0$. For complex search space with large numbers of local optima metaheuristic SA is capable of approximating the global optimum, due to its probabilistic acceptance of worse solutions which enables exploration beyond local minima~\cite{doi:10.1126/science.220.4598.671}. The pseudocode is presented in \algo{simulated-annealing}.

\begin{algorithm}[h]

	\KwIn{
    
    $\boldsymbol{A}$ : Input matrix of dimension $n\times n$;
    
$k$ : Dimension of submatrix;

$f$ : Objective function;

$N:$ Number of iterations;

$t:$ Initial temperature;

$\gamma$: Annealing parameter;}
	\KwOut{The submatrix $\boldsymbol{A}_B$ with the largest $f(\boldsymbol{A}_B)$; the function value $f(\boldsymbol{A}_B)$}
\vspace{1em}
        Generate a uniformly random binary vector $S$ of length $n$ with $k$ entries being one;

        Get the $k$-dimensional submatrix $\boldsymbol{A}_S$ according to $S$;

        Initialize $\boldsymbol{A}_B=\boldsymbol{A}_S$ and $Best=f(\boldsymbol{A}_S)$;
        
        Initialize $i=0$;
        
	\lWhile{$i < N$}\
    {
        \Indp Randomly generate an integer $m \in[0, k-1]$;

        Randomly choose $m$ nonzero entries of $S$ and keep them unchanged, then resample $k-m$ nonzero entries from the remaining $n-m$ vertices to get $R$; 

        Get $\boldsymbol{A}_{R}$ according to $R$;
        
        Calculate $f(\boldsymbol{A}_R)$;
        
        \lIf {$f(\boldsymbol{A}_R)>Best$}\
        \Indp $Best=f(\boldsymbol{A}_R)$;
        
        $\boldsymbol{A}_B=\boldsymbol{A}_R$;
        
        \Indm Set $S=R$ with probability $\min\{1,\exp\left[\frac{f\left(\boldsymbol{A}_R\right)-f\left(\boldsymbol{A}_S\right)}{t}\right]\}$
        
        \Indm \quad ~ $t~*= \gamma$
       
	\caption{Simulated Annealing}
     \label{algo:simulated-annealing}
 }
 
 Output $\boldsymbol{A}_B$ and $f(\boldsymbol{A}_B)$
\end{algorithm}

Finally, we substitute the uniformly random update in each iteration with \algo{glauber}, \algo{double-loop}, \algo{jerrum-glauber} and the quantum-inspired classical algorithm by~\citet{Oh_2024} again to enhance the above original simulated annealing. Specifically, we pick the initial $k$ vertices by Glauber dynamics, and in each iteration we remove the intersection of the Glauber dynamics outcomes with the $m$ unchanged vertices. Additionally, we remove a proper number of vertices randomly from the rest of the Glauber dynamics outcomes such that the number of the remaining vertices plus that of unchanged vertices equals $k$. The pseudocode is presented in \algo{simulated-annealing-enhanced}.

\begin{algorithm}[h]

	\KwIn{
    
    $\boldsymbol{A}$ : Input matrix of dimension $n\times n$;
    
$k$ : Dimension of submatrix;

$f$ : Objective function;

$N:$ Number of iterations;

$t:$ Initial temperature;

$\gamma:$ Annealing parameter;}
	\KwOut{The submatrix $\boldsymbol{A}_B$ with the largest $f(\boldsymbol{A}_B)$; the function value $f(\boldsymbol{A}_B)$}
\vspace{1em}
        Generate a binary vector $S$ of length $n$ with $k$ entries being one according to \algo{glauber} or \algo{double-loop} or \algo{jerrum-glauber};

        Get a $k$-dimensional submatrix $\boldsymbol{A}_S$ according to $S$;

        Initialize $\boldsymbol{A}_B=\boldsymbol{A}_S$ and $Best=f(\boldsymbol{A}_S)$;
        
        Initialize $i=0$;

	\lWhile{$i < N$}\
    {
        \Indp Randomly generate an integer $m \in[0, k-1]$;

        Generate $s$ by randomly choosing $m$ nonzero entries of $S$ and set the rest entries as zeroes;

        Get $\boldsymbol{A}_{{R}}$ according to $R$;

        Generate $r$ with $k$ entries are nonzero according to \algo{glauber} or \algo{double-loop} or \algo{jerrum-glauber} and set $m$ entries as zeroes randomly under the condition that there are no overlap between nonzero elements of $s$ and $r$;

        $R=s|r$;
        
        Calculate $f(\boldsymbol{A}_R)$;
        
        \lIf {$f(\boldsymbol{A}_R)>Best$}\
        \Indp $Best=f(\boldsymbol{A}_R)$;
        
        $\boldsymbol{A}_B=\boldsymbol{A}_R$;
        
        \Indm Set $S=R$ with probability $\min\{1,\exp\left[\frac{f\left(\boldsymbol{A}_R\right)-f\left(\boldsymbol{A}_S\right)}{t}\right]\}$
        
        \Indm \quad ~ $t~*= \gamma$

 }
 
 Output $\boldsymbol{A}_B$ and $f(\boldsymbol{A}_B)$

 \caption{Simulated Annealing Enhanced by Glauber Dynamics}
    \label{algo:simulated-annealing-enhanced}
\end{algorithm}

\end{document}